\documentclass[a4paper, reqno, 11pt]{amsart}
\usepackage{amsfonts, amsmath, amssymb, amsthm, amscd, setspace, subfigure,mathtools,stmaryrd}
\usepackage{verbatim,mdframed,graphicx,enumitem,cancel,mathrsfs}
\usepackage{pictexwd,dcpic,slashed,fancyhdr,float,array}
\usepackage[heightrounded,textwidth=400pt,top=1in,left=1in,right=1in,bottom=1in]{geometry}
\usepackage{layout}
\usepackage[utf8]{inputenc}
\usepackage[dvipsnames]{xcolor}
\usepackage[all,cmtip]{xy}
\usepackage{xr-hyper}
\usepackage[unicode,pagebackref=true]{hyperref}
\usepackage{cite}
\hypersetup{%
  linkcolor=NavyBlue,
  citecolor=ForestGreen,
  urlcolor=OrangeRed,
  anchorcolor=OrangeRed,
  colorlinks=true
}
\renewcommand*{\backref}[1]{}
\renewcommand*{\backrefalt}[4]{%
  \ifcase #1%
  \or [Page~#2.]%
  \else [Pages~#2.]%
  \fi%
}
\numberwithin{equation}{section}
\newtheorem{definition}{Definition}
\newtheorem{proposition}[definition]{Proposition}

\newcommand{\RR}{\mathbb{R}}
\newcommand{\ZZ}{\mathbb{Z}}

\newcommand{\so}{\mathfrak{so}}
\renewcommand{\sp}{\mathfrak{sp}}
\newcommand{\osp}{\mathfrak{osp}}
\newcommand{\g}{\mathfrak{g}}
\newcommand{\AdS}{\operatorname{AdS}}
\newcommand{\NW}{\operatorname{NW}}
\newcommand{\Cl}{\operatorname{Cl}}
\newcommand{\End}{\operatorname{End}}
\newcommand{\ddt}{\left. \frac{\mathrm{d}}{\mathrm{d}t} \right|_0}

\begin{document}

\title[Kaluza--Klein reductions of supersymmetric five-dimensional spacetimes]{Kaluza--Klein reductions of maximally supersymmetric five-dimensional lorentzian spacetimes}
\author[Figueroa-O'Farrill]{Jos\'e Figueroa-O'Farrill}
\author[Franchetti]{Guido Franchetti}
\address[JMF]{Maxwell Institute and School of Mathematics, The University
  of Edinburgh, James Clerk Maxwell Building, Peter Guthrie Tait Road,
  Edinburgh EH9 3FD, Scotland, United Kingdom}
\email{\href{mailto:j.m.figueroa@ed.ac.uk}{j.m.figueroa[at]ed.ac.uk}, ORCID: \href{https://orcid.org/0000-0002-9308-9360}{0000-0002-9308-9360}}
\address[GF]{Department of Mathematical Sciences, University of Bath,
  Claverton Down, Bath BA2 7AY, England, United Kingdom}
\email{\href{mailto: gf424@bath.ac.uk}{gf424[at]bath.ac.uk}, ORCID: \href{https://orcid.org/0000-0002-1511-6204}{0000-0002-1511-6204}}
\begin{abstract}
  A recent study of filtered deformations of (graded subalgebras of)
  the minimal five-dimensional Poincar\'e superalgebra resulted in two
  classes of maximally supersymmetric spacetimes.  One class are the
  well-known maximally supersymmetric backgrounds of minimal
  five-dimensional supergravity, whereas the other class does not seem
  to be related to supergravity.  This paper is a study of the
  Kaluza--Klein reductions to four dimensions of this latter class of
  maximally supersymmetric spacetimes.  We classify the lorentzian and
  riemannian Kaluza--Klein reductions of these backgrounds, determine
  the fraction of the supersymmetry preserved under the reduction and
  in most cases determine explicitly the geometry of the
  four-dimensional quotient.  Among the many supersymmetric quotients
  found, we highlight a number of novel non-homogeneous
  four-dimensional lorentzian spacetimes admitting $N=1$
  supersymmetry, whose supersymmetry algebra is not a filtered
  deformation of any graded subalgebra of the four-dimensional $N=1$
  Poincar\'e superalgebra.  Any of these four-dimensional lorentzian
  spacetimes may serve as the arena for the construction of new
  rigidly supersymmetric field theories.
\end{abstract}
\thanks{EMPG-22-14}
\maketitle
\tableofcontents

\section{Introduction}

It has been just over half a century since the emergence of
four-dimensional supersymmetry in a paper \cite{Golfand:1971iw} of
Golfand and Likhtman, which displayed for the first time what is now
called the four-dimensional $N=1$ Poincar\'e superalgebra. This is a Lie
superalgebra $\g = \g_{\bar0} \oplus \g_{\bar1}$, where the even
subalgebra $\g_{\bar0}$ is isomorphic to the Poincar\'e algebra and the
odd subspace $\g_{\bar1}$ is isomorphic to the four-dimensional real
spinor representation of $\g_{\bar0}$. The notation reflects the fact
that $\bar0$ and $\bar1$ are the residue classes modulo $2$ of $0$ and
$1$, respectively, in other words, the parity, and the Lie brackets
respect the parity. The Poincar\'e superalgebra also admits a compatible
$\ZZ$-grading $\g = \g_0 \oplus \g_{-1} \oplus \g_{-2}$, where $\g_0$
is isomorphic to the Lorentz algebra, $\g_{-2}$ is the translation
ideal and $\g_{-1} = \g_{\bar 1}$ is again the four-dimensional real
spinor representation.  Notice that the parity is simply the reduction
modulo $2$ of the grading, so that $\g_{\bar 0} = \g_0 \oplus
\g_{-2}$. This grading is an essential ingredient in the story, as we
will briefly explain below.

It did not take long for Zumino \cite{Zumino:1977av} to exhibit another
four-dimensional supersymmetry algebra extending this time the
isometry algebra of anti~de Sitter spacetime ($\AdS_4$).  That Lie
superalgebra, $\g \cong \osp(1|4)$, with $\g_{\bar0} \cong \so(3,2)
\cong \sp(4,\RR)$ and $\g_{\bar1}$ the four-dimensional vector
representation of $\sp(4,\RR)$, can be contracted \`a la In\"on\"u--Wigner
to the $N=1$ Poincar\'e superalgebra.

For many years, the only (lorentzian, $4$-dimensional) spacetimes
known to admit $N=1$ (i.e., $n:=\dim \g_{\bar1} = 4$) rigid (i.e.,
the metric is fixed and not a dynamical field) supersymmetry were
Minkowski and anti~de~Sitter spacetimes.  If we drop the condition
that the spacetime be lorentzian, perhaps by allowing it to be a
kinematical spacetime \cite{MR0238545, MR857383,
  Figueroa-OFarrill:2018ilb}, then one can obtain many supersymmetry
algebras (see, e.g., \cite{Puzalowski:1978rv, Palumbo:1978gx,
  Clark:1983ne, deAzcarraga:1991fa, MR769149, MR1723340,
  CampoamorStursberg:2008hm, Huang:2014ega}) extending the subclass of
kinematical Lie algebras of Bacry and L\'evy-Leblond \cite{MR0238545}
which can be obtained by contracting the isometry algebra of $\AdS_4$.
In fact, one can obtain a classification
\cite{Figueroa-OFarrill:2019ucc} of spatially isotropic kinematical
supersymmetry algebras (with four real superchanges) and their
corresponding homogeneous superspaces.

If we insist on spacetimes being lorentzian (spin) manifolds, then the
first four-dimensional examples beyond Minkowski and anti~de~Sitter
spacetimes were constructed by Festuccia and Seiberg
\cite{Festuccia:2011ws} via a decoupling limit of the theory obtained
by coupling matter (in the form of a sigma model) to off-shell
supergravity and then freezing the gravitational degrees of freedom
via a limit in which the Planck mass goes to infinity, resulting in a
background admitting rigid supersymmetry.  The supersymmetry here is
generated by the supergravity Killing spinors: sections of the spinor
bundle which are parallel relative to a connection constructed out of
the bosonic fields in the supergravity multiplet and uniquely
specified by the requirement of local supersymmetry.  Let us
elaborate briefly.

A supergravity theory is a gauge theory of supersymmetry, where the
gauge parameter is a section of a spinor bundle. The corresponding
gauge field in supergravity is called the gravitino and just as in
standard gauge theory it transforms infinitesimally as the
gauge-covariant derivative of the gauge parameter. The form of this
gauge-covariant derivative depends on the supergravity theory in
question and it is the covariant derivative associated to the
so-called gravitino connection on the spin bundle. The gravitino
connection may always be written as the spin connection modified by a
one-form with values in spinor endomorphisms, which depends a priori
on all the fields of the supergravity theory. One typically restricts
attention to so-called bosonic supergravity backgrounds: those for
which the fermionic fields are zero. Due to the nature of
supersymmetry transformations, any bosonic field configuration in such
a background is automatically invariant (infinitesimally) under a
supersymmetry transformation, but this is not necessarily the case for
the fermionic fields. On a bosonic supergravity background, the
gravitino transforms as the covariant derivative of the spinor gauge
parameter. Hence the supersymmetries which preserve the background are
generated infinitesimally by spinors which are parallel relative to
the gravitino connection. These are the supergravity Killing spinors.
If the supergravity theory contains other fermionic fields besides the
gravitino, their supersymmetric variations give rise to additional
algebraic equations Killing spinors must satisfy.

As shown in the eleven-dimensional context
\cite{Figueroa-OFarrill:2016khp} (but true in any dimension) Lie
superalgebras generated by parallel spinors (relative to \emph{any}
connection on the spinor bundle) have a particular algebraic
structure: they are filtered deformations of a graded subalgebra of
the Poincar\'e superalgebra and can be classified by generalised Spencer
cohomology \cite{Cheng:1999cy}. Moreover, the form of the putative
gravitino connection can be recovered from the cohomology
calculation. This then suggests a method of classification, whose
first step is the classification of graded subalgebras of the Poincar\'e
superalgebra.  That is not an easy problem, but one can nevertheless
classify \cite{deMedeiros:2016srz} filtered deformations of maximally
supersymmetric graded subalgebras; that is, filtered deformations of
graded subalgebras $\g = \g_0 \oplus \g_{-1} \oplus \g_{-2}$ with
$\dim \g_{-1}$ equal to the rank of the spinor bundle, which is $4$ in
this case. This recovers the supersymmetric spacetimes of Festuccia
and Seiberg \cite{Festuccia:2011ws} with one addition: a conformally
flat Cahen--Wallach \cite{MR267500} lorentzian symmetric space, whose
metric can be recognised as a bi-invariant metric on the Nappi--Witten
group \cite{Nappi:1993ie}.

This still leaves open the possibility of obtaining Lie superalgebras
$\g = \g_{\bar0} \oplus \g_{\bar1}$ with $\dim \g_{\bar1} = 4$ which
are filtered deformations of graded subalgebras of extended (i.e.,
$N>1$) Poincar\'e superalgebras.  These are supersymmetric extensions of
the Poincar\'e algebra, where the odd subspace is now $N$ copies of the
$4$-dimensional spinor representation.  Some of these extended
Poincar\'e superalgebras can be obtained by dimensional reduction from
minimal Poincar\'e superalgebras in higher dimensions and all we need to
do is break some of the supersymmetry.  This is typically achieved via
Kaluza--Klein reduction of higher dimensional supersymmetric
backgrounds.  Geometrically this corresponds to viewing the
higher-dimensional background as a principal bundle over the
four-dimensional background: the surviving four-dimensional
supersymmetry is then generated by those higher-dimensional Killing
spinors which are invariant under the structure group of the principal
bundle.  Spencer cohomology calculations in five
\cite{Beckett:2021cwx} and six \cite{deMedeiros:2018ooy} dimensions
have resulted in a list of supersymmetric backgrounds with eight real
supercharges and, intriguingly, not all of them seem to be related to
supergravity.

In this paper we concentrate on the five-dimensional supersymmetric
backgrounds described in \cite{Beckett:2021cwx}.  In that paper,
following from the calculation of the generalised Spencer cohomology
of the minimal five-dimensional Poincar\'e superalgebra, two families of
maximally supersymmetric backgrounds were obtained.  One family
consists of the maximally supersymmetric backgrounds of minimal
five-dimensional supergravity \cite{Gauntlett:2002nw}, and the other
family consists of five-dimensional lorentzian locally symmetric
spaces $(M,g)$ together with a nontrivial $\sp(1)$-valued vector field
$\Phi =\varphi \otimes r $, with $r$ a fixed element of $\sp(1)$ and
$\varphi \in \mathfrak{X}(M)$ a parallel vector field.  Up to
covering, $(M,g)$ is one of the following (conformally flat) spacetimes
\begin{enumerate} 
\item $-\mathbb{R} \times S^4$ with the lorentzian product metric, if
  $\varphi $ is timelike;
\item $\RR \times \AdS_4$ also with the lorentzian product metric, if
  $\varphi $ is spacelike; and
\item a certain indecomposable Cahen--Wallach (CW) lorentzian symmetric
  space, if $\varphi$ is null.
\end{enumerate}
In the first two cases, the lorentzian norm of $\varphi$ is related to
the curvature of the round metric on $S^4$ and on $\AdS_4$.

The aim of this paper is to classify four-dimensional Kaluza--Klein
reductions of these three backgrounds and in particular those which
preserve some (nonzero) fraction of the supersymmetry. For each space
we will proceed in a similar way: first we classify one-parameter
subgroups of the relevant isometry group which lead to quotients which
are smooth pseudo-riemannian manifolds; that is, either lorentzian or
riemannian. We then find which of these one-parameter subgroups
preserve some fraction of supersymmetry. Finally we discuss the
geometry of these quotients and comment on the results obtained.

The paper is organised as follows.  In Section~\ref{setup} we set out
our conventions, introduce the five-dimensional geometries and discuss
our methodology: how to classify the one-parameter subgroups of
isometries and how to select those which lead to either a lorentzian
or riemannian quotient; how to determine the fraction of supersymmetry
which is preserved by the reduction; and how to work out the quotient
metric and its  isometries.  We then apply this methodology to each of
the three five-dimensional spacetimes in turn.

In Section~\ref{s4} we study the reductions of $-\RR \times S^4$. The
one-parameter subgroups are determined in Proposition~\ref{kvfsn} and
the corresponding fraction of supersymmetry is determined in
Proposition~\ref{psusys4}. Among the reductions, it is worth
mentioning a one-parameter family of $\tfrac12$-BPS quotients (i.e.,
preserving one half of the supersymmetry) which are riemannian
four-dimensional manifolds admitting rigid euclidean supersymmetry
with four real supercharges. These manifolds are diffeomorphic to
$S^4$ but the metric is far from the round metric, being of
cohomogeneity one with $O(3) \times O(2)$ isometry group.

In Section~\ref{ads4} we study the reductions of $\RR \times \AdS_4$.
The one-parameter subgroups are determined in
Proposition~\ref{kvfads} and the corresponding reductions contain both
lorentzian and riemannian four-dimensional geometries.  Those
one-parameter subgroups preserving some supersymmetry are determined
in Proposition~\ref{psusyads4} and listed in Table~\ref{presusyads}.
We see that there are two $\tfrac12$-BPS reductions: one riemannian
and one lorentzian.  The metric of the lorentzian reduction is given
in equation~\eqref{eq:new-lor-4d-metric-ads} and that of the riemannian
reduction by equation~\eqref{eq:new-riem-4d-metric-ads}.  The former
lorentzian metric is not conformally flat and hence it is not one of
the four-dimensional geometries in \cite{deMedeiros:2016srz}.  It is a
novel four-dimensional lorentzian geometry admitting rigid
supersymmetry with four real supercharges.  The latter riemannian
metric gives a four-dimensional riemannian geometry admitting rigid
euclidean supersymmetry with four real supercharges.  The (hereditary) isometry Lie
algebras of the quotients are listed in Table~\ref{adsgeneratorslie}.
We observe that, in particular, the isometry Lie algebra of the
lorentzian metric in equation~\eqref{eq:new-lor-4d-metric-ads} is
isomorphic to $\so(2) \oplus \so(2,1)$, while that of the riemannian
metric in equation~\eqref{eq:new-riem-4d-metric-ads} is isomorphic to
$\so(2) \oplus \so(3)$.  Neither metric is homogeneous, but of
cohomogeneity one.

In Section~\ref{cw} we study the reductions of the Cahen--Wallach
spacetime. The possible Killing vector fields are determined in
Proposition~\ref{kvfcw}. All the resulting quotients are lorentzian.
Those quotients preserving some supersymmetry are determined in
Proposition~\ref{psusycw} and listed, along with the corresponding
fraction of supersymmetry, in Table~\ref{frusycw}. We see that there
are three families of reductions preserving half the supersymmetry.
Upon determining the quotient metric, we see that two of them
(labelled $X_6$ and $X_8^\pm$) are isometric to the bi-invariant metric on
the Nappi--Witten group, hence they agree with one of the geometries
in \cite{deMedeiros:2016srz}. The third family of reductions (labelled
$X_9$) contains at least a one-parameter family of novel four-dimensional lorentzian
geometry admitting rigid $N=1$ supersymmetry. The (hereditary) isometry Lie
algebras of the quotients are listed in Table~\ref{geoquotcw}.

Finally in Section~\ref{sec:conc-summ} we summarise the results and
offer some conclusions.  The paper contains two appendices:
Appendix~\ref{cwgener} on the geometry of the Cahen--Wallach spaces
and Appendix~\ref{gmatr} on the explicit form of the gamma matrices
used to check the fraction of supersymmetry in the quotients.

\section{Setup}
\label{setup} 
\subsection{Conventions}

We will work with metrics $g$ of mostly plus signature. Accordingly we
call a nonzero vector $v$ timelike if $g (v,v )<0$, spacelike if
$g(v,v)>0$ and null if $g (v,v )=0$.  We denote by $- \mathbb{R}$
the space $\mathbb{R} $ with negative-definite Euclidean inner
product, by $\mathbb{R}^{p,q}$ the space $\mathbb{R}^{p+q}$
with flat pseudo-inner product of signature $(p,q) $, where $p$ is the
number of negative eigenvalues, and by $\eta^{p,q}$, or just $\eta$
if the signature is clear from the context, the corresponding inner
product. The isometry group of a pseudo-riemannian manifold $M$ is
denoted by $ \operatorname{Iso} (M)$. Einstein's summation convention
is enforced throughout.

Let $V$ be a finite dimensional vector space $V$ equipped with a
non-degenerate inner product $g$. The Clifford algebra $\Cl (V) $
associated to $(V,g) $ is defined by the relation
\begin{equation} 
\label{clifrel} 
u v + vu = - 2 g (u,v ) I 
\end{equation} 
for any $u,v \in V $.  For $ (p,q )=(1,4 ) $,
$ \Cl (V) \simeq \End(\Sigma) \oplus \End(\Sigma^\prime)$ where $\Sigma$,
$\Sigma^\prime$ are irreducible inequivalent $\Cl (V) $ modules, each
of which is isomorphic to $\mathbb{H}^2 $ as a (right) vector
space. They are distinguished by the action of the volume element
which is by $I$ on $\Sigma$ and $- I $ on $\Sigma^\prime $.  We will
work with $\Sigma$ from now on.  The Clifford algebra $\Cl (V) $ is
generated by an orthonormal basis $(e_i )$ of $V$, and we denote by
$\gamma_i $ the image of $e_i $ under the representation
$\Cl (V) \rightarrow \operatorname{End} \Sigma $.  The gamma matrices
obey the relation
\begin{equation}
\gamma_i \gamma_j + \gamma_j \gamma_i =-2 \eta^{ 1,4 }_{ ij }.
\end{equation} 
We also define
\begin{equation}
  \gamma_{ ij } = \frac{1}{2} (\gamma_i  \gamma_j - \gamma_j \gamma_i ).
\end{equation}

The Riemann tensor $R$ of a pseudo-riemannian manifold $(M,g) $ is
\begin{equation}
R(X,Y) Z = (\nabla_X \nabla_Y - \nabla_Y \nabla_X - \nabla_{ [ X , Y ]} ) Z,
\end{equation} 
with components with respect to a local frame $\{ e_i \} $ and dual coframe $\{ e^i \} $
\begin{equation}
R^i_{  \phantom{ i } jkl } = e^i ( R ( e_k , e_l )  e_j).
\end{equation} 
The associated Ricci tensor $ \mathrm{Ric}$ and scalar curvature $s$ are
\begin{equation}
\mathrm{Ric} _{ ij } = R^k_{ \phantom{ k } i k j }, \quad 
s =\mathrm{Ric}^i_{ \phantom{ i }i }.
\end{equation} 
We also define
\begin{equation}
R_{ i jkl } = g_{ i m } R^m_{ \phantom{ m }jkl }.
\end{equation} 

\subsection{Maximally supersymmetric five-dimensional backgrounds}

As mentioned in the Introduction, we are going to consider
four-dimensional Kaluza--Klein (KK) reductions of maximally
supersymmetric backgrounds associated to filtered deformations of
graded subalgebras $\g = \g_{ 0} \oplus \g_{-1} \oplus \g_{-2}$ of the
minimal five-dimensional Poincar\'e superalgebra with $\dim\g_{-1} = 8$
(and hence with $\dim\g_{-2} = 5$). It is proved in
\cite{Beckett:2021cwx} that these are either the usual maximally
supersymmetric backgrounds of minimal supergravity, or locally
symmetric spaces $(M,g) $ with Riemann tensor
\begin{equation}
\label{riemtens} 
R_{ i  j k l } 
= - \tfrac{1}{2} \left( 
g_{i  l  }\varphi_j  \varphi_k- g_{i  k }\varphi_j  \varphi_l 
+  g_{j k }\varphi_i   \varphi_l - g_{j l  }\varphi_i   \varphi_k
   +  |\varphi |^2  \left( g_{ i  k } g_{ j l } -g_{ i  l }g_{ j k } 
 \right)
   \right),
\end{equation} 
where $\varphi\in\mathfrak{X}(M)$ is a certain parallel vector field.
The corresponding Ricci tensor and scalar curvature are
\begin{equation}
\mathrm{Ric} _{ij }=\tfrac{3}{2} ( \varphi_i \varphi_j - |\varphi |^2 g_{ ij } ), \quad 
s =-6 |\varphi |^2 .
\end{equation} 
It follows that, up to covering, $(M,g) $ is one of of three possible
spaces depending on the causal character of $\varphi $.

If $\varphi $ is timelike then  $ M =-\mathbb{R} \times S^4$ with the lorentzian product metric. 
 We take global coordinates $(x^0, x^1 , x^2 , x^3 , x^4  , x^5
 ) $ on $  \mathbb{R} ^{1, 5}  $ and embed $S^4  \subset \mathbb{R}
^5 $ in the usual way, $S^4 = \{ (x^1 , \ldots , x^5 ): x_1^2 +
 \cdots + x_5^2 =R^2  \} $.
In this case
\begin{equation}
\varphi =c \partial_0 .
\end{equation} 
 By (\ref{riemtens}), the radius $R$  of $S^4 $ is related to $|\varphi |$ by the equation
 \begin{equation}
R^2 |\varphi |^2 =-2.
 \end{equation} 

If $\varphi$ is spacelike then  $M =\mathbb{R}  \times \mathrm{AdS}_4 $ where $\mathrm{AdS}_4 $ is 4-dimensional anti de Sitter space.
Taking global coordinates $(y , x^1 , x^2 , x^3 , x^4 ,  x^5 ) \in \mathbb{R} \times \mathbb{R} ^{ 2,3 } $,  with $x^1 , x^2 $ timelike and the other coordinates spacelike, we regard $\mathrm{AdS}_4 $ as the universal  cover of the quadric in $\mathbb{R} ^{ 2,3 }$
\begin{equation}
x_1^2 +  x_2^2 - x_3^2 -x_4^2 - x_{5}^2 = R^2 .
\end{equation} 
In this case
\begin{equation}
\varphi =c \partial_y 
\end{equation} 
and by (\ref{riemtens}) the ``radius'' $R$  of $\mathrm{AdS}_4 $ is related to $|\varphi |$ by the equation
 \begin{equation}
R^2   |\varphi |^2 =2.
 \end{equation} 

If $\varphi$ is null then $M $ is an indecomposable Cahen--Wallach (CW) space. Some general facts about Cahen--Wallach spaces are recalled in Appendix \ref{cwgener}.  In five dimensions the geometry is that of  $\mathbb{R} ^5 $ with global coordinates $(x^+ , x^- , x^1 , x^2 , x^3) $ and the lorentzian metric
\begin{equation}
\label{cwgeneral} 
g =2 \mathrm{d} x^+ \mathrm{d} x^-  + \sum_{ i , j =1 }^3 A_{ ij } x^i x^j   (\mathrm{d} x^- )^2  +  | \mathrm{d} x |^2 ,
\end{equation} 
where $|\mathrm{d} x |^2 = \mathrm{d} x_1^2 + \mathrm{d} x_2^2 + \mathrm{d} x_3^2 $ and $A$ is a bilinear symmetric form on $\mathbb{R} ^3 $.  The only non-vanishing components of the Riemann and Ricci tensors associated to (\ref{cwgeneral}) are
\begin{equation}
\label{riemw} 
R_{ -i-j }=-  A_{ ij }, \quad 
 \operatorname{Ric} _{ -- } = - \operatorname{Tr} A .
\end{equation} 
In this case
\begin{equation}
\varphi = c \partial_+ .
\end{equation} 
Comparing (\ref{riemw}) with (\ref{riemtens}) we  find  that for the CW spaces arising in \cite{Beckett:2021cwx}  $A$ is a scalar matrix,
\begin{equation}
A  = -\frac{c^2 }{2} \eta^{ 0,3 }.
\end{equation} 
For such a choice of $A$ the metric (\ref{cwgeneral}) is conformally flat, see (\ref{weylcw}).
From now on  we take $ c =\sqrt{ 2 } $ so that
 \begin{equation}
g =2 \mathrm{d} x^+ \mathrm{d} x^-  -  |x |^2    (\mathrm{d} x^- )^2  +  | \mathrm{d} x |^2 ,
\end{equation} 
where $|x |^2 =x_1^2 + x_2^2 + x_3^3 $.

\subsection{One parameters subgroups}

Let $(M,g) $ be a pseudo-riemannian manifold,
$G= \operatorname{Iso} (M)$. For $\Gamma \subset G $ a suitable
one-parameter subgroup of $G$, the Kaluza--Klein (KK) reduction of
$(M,g) $ along $\Gamma$ is the space of orbits $M/ \Gamma $ with the
induced metric. In order for $M / \Gamma $ to have a well defined
pseudo-riemannian metric we require $\Gamma$ to be generated by a
Killing Vector Field (KVF) $\xi$ which is either spacelike or timelike, so
$|\xi |\neq 0 $. The quotient by a null Killing vector field, leading
to Newton-Cartan geometries, is also interesting but we defer it to
future work.

The condition $|\xi |\neq 0 $ automatically ensures that the action
generated by $\xi$ is free. In order to ensure that $M$ is a smooth
manifold we also want $\Gamma$ to act properly.  In indefinite
signature determining if the action of a subgroup of $G$ is proper is
a non-trivial problem which we will not address in this paper, but
simply assume that parameters are chosen in such a way that the action
is proper.

Let $\Gamma $ be a one-parameter subgroup of $G$, $g \in G $ an
isometry. In classifying KK reductions we want to identify the
isometric quotients
\begin{equation}
\label{isoquotients}
M /  \Gamma  \simeq  (g \cdot M) / (g\Gamma g ^{-1}).
\end{equation}
The $G$ action on $M$ induces a Lie algebra homomorphism $  \mathfrak{g} \rightarrow \mathfrak{X} (M)   $ taking values in the space of Killing vector fields, given by 
\begin{equation}
\label{liegkvf} 
X \mapsto  \xi_X , \quad  \xi_X |_p =\ddt \exp ( - t X ) \cdot p ,
\end{equation}
where we act with $- X $ rather than $X$ in order to obtain a Lie algebra homomorphism rather than anti-homomorphism. It follows that  if $ \xi_X $ is a Killing vector field and $\Gamma=\Gamma_{\xi_X} $  the corresponding one-parameter subgroup of $G$ then
\begin{equation}
g\Gamma_{\xi _X}g^{-1}=\Gamma_{\xi_{\operatorname{Ad}_gX}}.
\end{equation}
Hence for any $g \in G $, $X \in \mathfrak{g}  $, for $\Gamma=\Gamma_{\xi_X}$ (\ref{isoquotients}) is equivalent to
\begin{equation}
M /  \Gamma_{\xi_X}  \simeq  (g \cdot M) / \Gamma_{\xi_{\operatorname{Ad}_gX}}.
\end{equation}
Moreover for any  $\lambda  \in \mathbb{R} ^\times $, $X\in \mathfrak{g}  $ and $\lambda  X $ generate the same one-parameter subgroup. Therefore classifying one-parameter subgroups of $G$ is equivalent to classifying 1-dimensional Lie subalgebras  of $\mathfrak{g}  $  under the equivalence relation 
\begin{equation}
 X \sim \lambda \, g X g^{-1} ,
\end{equation} 
for any $ \lambda \in \mathbb{R} ^\times $, $g \in G $.
We also recall that, for any $g\in G$, $X\in \mathfrak{g}$,  (\ref{liegkvf})  satisfies
\begin{equation}
\label{liekv} 
\xi_{ \operatorname{Ad}_g X  }  = g_\ast  \, \xi_X ,
\end{equation} 
where $g_\ast $ is the induced $G$-action on $T M $.
Hence conjugation  of $X$ by $g$ corresponds to left translation  by an isometry of the associated KVF.

\subsection{Preserved SUSY}

As discussed in the Introduction, for the purposes of this paper a Killing spinor is a solution of the parallelism condition $\mathcal{D} \epsilon =0$ for $\mathcal{D}$ a certain connection on the spinor bundle which in the present case is \cite{Beckett:2021cwx}
\begin{equation}
\mathcal{D} = \nabla - \beta ,
\end{equation}
where $\nabla $ is (the lift to the spin bundle of) the Levi-Civita connection and, for $ \xi $ a Killing vector field, $\epsilon$ a spinor field, $\beta \in \Gamma  (T^\ast M  \otimes \operatorname{End} ( \Sigma  ) )$ is given by
\begin{equation} 
\label{xicontr} 
\beta (\xi  ) (\epsilon ) =: \beta_\xi  \epsilon    = \frac{1}{4} (\xi  \cdot \varphi \cdot + 3 g (\xi , \varphi ) I ) r \epsilon  
 = \frac{1}{4}  \xi ^i\varphi^j (\gamma_i \gamma_j + 3g_{ i j } ) r \epsilon.
\end{equation} 
Here $r$ is some fixed non-zero element of $\mathfrak{su }(2) $,
\begin{equation}
r =\begin{pmatrix}
i |r^1_{ \ 1 } |& r^1_{ \ 2 } \\
- \overline{r^1_{ \ 2 } } & - i |r^1_{ \ 1 } |
\end{pmatrix} 
\end{equation}  
which we are free to rescale  as well to conjugate within $SU (2) $. Hence from now on we take
\begin{equation}
\label{rchoice} 
r =\begin{pmatrix}
i & 0 \\
0 & -i
\end{pmatrix} .
\end{equation} 

A given bosonic background $M$ is supersymmetric if
$\nu = \operatorname{dim} (S) / n >0$ for $S$ the space of Killing
spinors on $M$ and $n$ the rank of the spinor bundle. If $\nu =1 $
then $M$ is said to be maximally supersymmetric. As discussed in
\cite{Figueroa-OFarrill:2004lpm}, if $M$ is simply connected spin,
$\xi$ a KVF, $\Gamma_\xi $ the corresponding one-parameter group, then
$M / \Gamma_\xi $ is a spin manifold if and only if and only if the
$\Gamma_\xi $-action on $M$ lifts to an action on the spin bundle of
$M$ in a $\Gamma_\xi $-equivariant way.\footnote{In
  \cite{Figueroa-OFarrill:2004lpm} only the case $\xi $
  spacelike is considered, but the argument extends directly to any
  KVF $\xi$ which never vanishes.}  This always happens if $\Gamma_\xi
$ has the topology of a line, which is the case for all the examples
we will consider.

In standard supergravity, the KK reduction $M / \Gamma_\xi $ of a
$5d$ solution gives a solution of the $4d$ supergravity equations
provided that $\xi$ preserves any other field which is turned on in
five dimensions: e.g., for minimal supergravity, $\xi$ should also
preserve the $2$-form.  A $4d$ Killing spinor lifts to a $5d$ Killing
spinor which is invariant under $\xi$.  Conversely, any $5d$ Killing
spinor which is invariant under $\xi$, descends to a Killing spinor on
$M / \Gamma_\xi $.  In our case we prefer not to specify a $4d$
theory on the quotient, but in analogy with standard supergravity we
define a Killing spinor on $M /\Gamma_\xi $ to be a Killing spinor on
$M$ which is invariant under $\Gamma_\xi $.

A Killing spinor $\epsilon$ is invariant under $\Gamma_\xi $ if and
only if $L_\xi \epsilon =0 $ for $L_\xi $ the spinorial Lie
derivative. The fraction of SUSY preserved by the quotient
$M / \Gamma_\xi $ is thus
 \begin{equation}
\frac{1}{n} \operatorname{dim}  \left(  \operatorname{Ker} L_\xi |_S  \right).
 \end{equation}  
We recall that the Lie derivative of a spinor $s$  along a KVF $\xi$ can be written
\begin{equation}
L_\xi s = \nabla_\xi s+ \frac{1}{4} \mathrm{d} \xi^\flat \cdot s .
\end{equation} 
For  a Killing spinor $\epsilon$, $\nabla \epsilon = \beta \epsilon $, so $\Gamma_\xi$ preserves a Killing spinor if
\begin{equation}
\label{lxie} 
L_\xi \epsilon =  \beta _ \xi \epsilon + \frac{1}{4} \mathrm{d} \xi^\flat \cdot \epsilon =0.
\end{equation} 
  Since $\xi$ preserves all the bosonic fields which are
  turned on, $L_\xi \circ\mathcal{D} =0$, hence $\xi$ preserves
  the spinor connection and $[L_\xi, \mathcal{D}_X] = \mathcal{D}_{[\xi,X]}$. Therefore the Lie derivative of a Killing spinor is a
Killing spinor. Being defined as parallel sections, Killing spinors (on
a connected manifold) are fully determined by their value at one
point, hence (\ref{lxie}) only needs to be checked at an arbitrarily
selected point, which we can choose so to simplify computations.

Writing a spinor $\epsilon$ as a pair $(\epsilon_1 , \epsilon_2 )^T $, with $\epsilon_i \in \mathbb{H}  \simeq \mathbb{R} ^4 $ acted upon by the gamma matrices $ \gamma_\mu \in \Cl (1,4) $, and with the choice (\ref{rchoice}) for $r$, (\ref{lxie}) decouples into the two matrix equations
\begin{equation}
\begin{split} 
\left[\frac{1}{8}(\mathrm{d} \xi^\flat )_{ \mu \nu}\gamma^{\mu \nu } +  \frac{ i}{4} \xi ^\mu\varphi^\nu (\gamma_\mu \gamma_\nu + 3g_{ \mu \nu } )\right] \epsilon_{1} &=0,\\
\left[\frac{1}{8}(\mathrm{d} \xi^\flat )_{ \mu \nu}\gamma^{\mu \nu } -\frac{ i}{4} \xi ^\mu\varphi^\nu (\gamma_\mu \gamma_\nu + 3g_{ \mu \nu } )\right] \epsilon_{2}&=0 .
\end{split} 
\end{equation}

\subsection{Geometry of the quotient}
\label{gqgd}

We first recall the geometry of a KK quotient.
Let $(M,g) $ be a pseudo-riemannian $5$-manifold  with a connected 1-dimensional Lie group $\Gamma $ acting smoothly and properly by isometries. Let  $\xi  $ be the KVF generating the $\Gamma$  action. Provided that $|\xi  |\neq 0 $ everywhere, the quotient $M/\Gamma  $ is a smooth pseudo-riemannian $4$-manifold and 
\begin{equation} 
\pi : M \rightarrow M /\Gamma  
\end{equation}
 a principal $\Gamma $-bundle. 
The metrics $h$ on $ M / \Gamma  $ and $g$ on $M$ are related by
\begin{equation}
\label{kkm} 
g = \pi^\ast h  +  \frac{ \xi ^\flat \otimes \xi ^\flat }{g (\xi  ,\xi  ) },
\end{equation} 
where $ \xi ^\flat =g (\xi  , \cdot )$.
In adapted local coordinates with $\xi  =\partial / \partial z $ we have
\begin{equation}
  \xi ^\flat = \pm \mathrm{e}^{ 2 \tilde\phi } ( \mathrm{d} z + A ),
\end{equation} 
with the top (respectively bottom) sign if $\xi $ is spacelike (timelike), $\tilde\phi \in C^\infty (M) $ and $A \in \Omega^1 (M )$ such that $\xi  (\tilde\phi )=0 $,  $i_\xi  A =0 =i_\xi  \mathrm{d} A $.
It follows that $\tilde\phi = \pi^\ast \phi$ for some $\phi \in
C^\infty(M/\Gamma)$ and that $dA = \pi^\ast F$ for some $F \in
\Omega^2(M/\Gamma)$.  In summary, (\ref{kkm}) can be  written locally as
\begin{equation}
\label{kkquotmetric} 
g = \pi^\ast h  \pm  \mathrm{e}^{ 2 \pi^\ast  \phi }(\mathrm{d} z + A )^2.
\end{equation} 

In all the cases that we are going to study $M$ is a trivial
principal bundle. In fact the KVFs considered in Sections \ref{s4},
\ref{ads4}, \ref{cw}, see Propositions \ref{kvfsn}, \ref{kvfads},
\ref{kvfcw}, all generate groups $\Gamma$ having the topology of a
line, and principal $\RR$-bundles over paracompact bases are always
trivial, see e.g. \cite[Proposition~I.32]{MR2000747}.

By making a suitable choice of coordinates, for example adapted to the
KVF $\xi$, we can explicitly exhibit the metric on $M / \Gamma $,
whose isometry group can then be easily determined, at least in some
cases.  We also recall that if $G$ acts effectively on $M$ by
isometries and $\Gamma$, $H $ are subgroups of $G$, then the $H$
action descends to $M /\Gamma$ if and only if $H$ is a subgroup of the
normaliser $N_\Gamma (G) $.  Any isometry of $M$ descending to $M/
\Gamma $ is an isometry of $M / \Gamma $. Of course $\Gamma$ itself
acts trivially on $M / \Gamma $, so we need to quotient by it to
obtain an effective action,
\begin{equation}
\label{isoquotform} 
N_\Gamma (G)/ \Gamma \subset\operatorname{Iso} ( M / \Gamma ).
\end{equation}
We call the Lie algebra of $ N_\Gamma (G)$ the hereditary isometry
algebra of the quotient and denote it by $\mathfrak{l }$. In general
it is a proper subalgebra as $ M / \Gamma $ may have additional
``accidental'' isometries.

\newpage 
\section{Kaluza-Klein quotients of $- \mathbb{R}  \times S^4 $}
\label{s4}

Recall that we take global coordinates $(x^0, x^1 , x^2 , x^3 , x^4  , x^5 ) $ on $  \mathbb{R} ^{1, 5}  $ with $S^4 = \{ (x^1 , \ldots , x^5 ): x_1^2 + \cdots + x_5^2 =R^2  \} $.  The parallel vector field $\varphi$ is 
\begin{equation}
\varphi =c \partial_0 ,
\end{equation} 
with 
 \begin{equation}
R^2 |\varphi |^2 = - R^2 c^2 =-2.
 \end{equation}

\subsection{One parameters subgroups of $G$}
\label{onepars4} 
The isometry group of $M =- \mathbb{R}  \times S^4 $ is the direct product $ G = \mathbb{R}  \times O (5) $ of time translations and rotations, with Lie algebra  $\mathfrak{g}  = \mathbb{R}  \oplus  \mathfrak{so }(5)$. Let $X = (\tau^0 , - \omega) \in \mathfrak{g}  $. Here  
\begin{equation}
 \omega = \omega^i_{ \ j } e_i \otimes e^j 
 \end{equation}
  is a skew-adjoint endomorphism, $\omega_{ ij }:=g_{ i a }\omega^a_{  \ j } =- g_{ ja }\omega^a_{ \ i } = - \omega_{ ji }$. The vector field $\xi  $ associated to $X$ by (\ref{liegkvf}) is, taking into account that the action on coordinates is the inverse of that on points,
\begin{equation}
\label{xitoxxxxx} 
\begin{split}
\xi_p &
  = \ddt x^\mu ( \exp (-t X ) \cdot p ) \partial_\mu 
    = \ddt (\exp (t X ) )^\mu_{ \ \nu }x^\nu  ( p ) \partial_\mu \\ &
= \ddt \Big( (x^0 (p)  +  \mathrm{e}^{t  \tau^0} )\partial_0  +  (\mathrm{e}^{ -t \omega } )^i_{ \ j } x^j(p)   \partial_i\Big )
= \tau^0  \partial_0 -\omega^i_{ \ j } x^j (p) \partial_i \\ &
= \tau^0\partial_0  + \frac{1}{2}  \omega ^{ i j } ( x_i (p) \partial_j - x_j (p) \partial_i ).
\end{split}
\end{equation}
Identifying a point with its coordinates, from now on we simply write
\begin{equation}
\xi =\tau + \lambda , \quad   \tau = \tau^0 \partial_0,\quad  
 \lambda = \tfrac{1}{2} \omega ^{ ij } (x_i \partial _j - x_j \partial_i ).
\end{equation}
Therefore we have the correspondence
\begin{equation}
\label{lieveccorr} 
(\tau ^0 , - \omega ) \in \mathfrak{g}  \quad  \longleftrightarrow \quad   \tau^0 \partial_0 + \frac{1}{2} \omega^{ ij } R_{ ij } \in \mathfrak{X}  (M) ,
\end{equation} 
where\begin{equation}
\begin{split} 
R_{ ij } &= x_i \partial_j -x_j  \partial_i 
 \end{split}
 \end{equation} 
is the generator of a  rotation in the plane $(i,j) $. Its  squared norm with respect to $\eta^{ 1,5 }$ is
\begin{equation} 
\label{rotnorm} 
|R_{ ij } |^2 =x_i^2 + x_j^2 .
\end{equation} 
In particular, let $(e_i )  $, $(e^i )$ be bases of $\mathbb{R}  ^5 $ and of its dual.  Taking generators
\begin{equation}
\label{genso23} 
\epsilon_{ ij } = e_j \otimes e^i  - e_i \otimes e^j , \quad i, j =1, \ldots ,5,
\end{equation} 
of $\mathfrak{so }(5) $, with bracket
\begin{equation}
[\epsilon_{ ij }, \epsilon_{ kl }]= 
\eta^{ 0,5 }_{ ik } \epsilon_{ jl }+\eta^{ 0,5 }_{ jl } \epsilon_{ ik } -\eta^{ 0,5 }_{ il } \epsilon_{ jk } - \eta^{ 0,5 }_{ jk }\epsilon_{ il },
\end{equation}  
we have
\begin{equation}
\label{partcorrrlkj} 
\epsilon_{ ij }\in \mathfrak{so }(5)  \quad \longleftrightarrow \quad R_{ ij } \in \mathfrak{X}  (S^4).
\end{equation} 

\begin{proposition} 
\label{kvfsn} 
Let $\xi$ be a KVF of $M = - \mathbb{R}  \times S  ^{4} $ and assume that $|\xi | $ never vanishes. Then $\xi$ is timelike and there are coordinates such that, up to rescaling,
\begin{equation}\label{kvfs4main} 
  \xi =\partial_0 + \lambda , \quad \lambda  =  \beta_1 R_{ 12 }   + \beta_2  R_{ 34 },
\end{equation} 
with $\beta_1 $, $\beta_2 $ satisfying
\begin{equation}\label{sn23445} 
  R |\beta _1| < 1 \qquad\text{and}\qquad   R |\beta_2| < 1.
\end{equation}
\end{proposition} 

\begin{proof} 
Since rotations fix the origin of $\mathbb{R} ^{ 1,5 }$, in order for $\xi $ not to have zeros we need $\tau \neq 0$ and we can always rescale  $\xi $ so that $\tau =\partial_0 $. As it is well known, by conjugating via a rotation $\lambda$   can be brought to normal form
\begin{align} 
\label{fixsn} 
\lambda  &= 
 \beta_1 R_{ 12 }   + \beta_2  R_{ 34 },
\end{align} 
so that $|\xi |^2 = -1 + |\lambda |^2  = -1 + \beta_1^2 (x_1^2 + x_2^2 ) + \beta_2^2 (x_3^2 + x_4^2 )$. Restricting to $S^4 $  we have
\begin{equation*} 
0 \leq |\lambda |^2 \leq \operatorname{max} \{ \beta_1^2R^2, \beta_2^2 R^2   \}.
\end{equation*} 
Hence in order to avoid zeros of $|\xi  |$  we need  $ \operatorname{max} \{ \beta_1^2 R^2 , \beta_2^2 R^2   \}  < 1  $ and $\xi $ is timelike.
 \end{proof} 
 Note that since (\ref{kvfs4main}) always has a non-trivial translation part, the associated one-parameter group is non-compact and has the topology of a line.
 
\subsection{Preserved SUSY}
\begin{proposition}\label{psusys4}
  The KVF 
  \begin{equation}
    \xi =\partial_0 +  \beta_1 R_{ 12 }   + \beta_2  R_{ 34 }
  \end{equation} 
  of Proposition \ref{kvfsn} preserves some SUSY if and only if 
  \begin{equation}
    (\beta_1 \pm  \beta_2 )^2 =- |\varphi |^2 .
  \end{equation}
  The fraction $\nu$  of  preserved SUSY  is 
  \begin{equation*}
    \nu =\begin{cases} 
      \frac{1}{2} &\text{if $\beta_2 \beta_1 =0 $},\\
      \frac{1}{4} &\text{otherwise} .
    \end{cases} 
  \end{equation*} 
\end{proposition} 
\begin{proof} 
By the Poincar\'e--Hopf theorem any vector field on $S^4 $ has a zero,
so let $o$ be a zero of $\lambda =\beta_1 R_{ 12 }+ \beta_2 R_{ 34 }$.
Then $\xi_o =\partial_0 |_o $ and substituting
\begin{equation}
 \varphi = c  \partial_0
 \end{equation}
 in (\ref{xicontr}) and evaluating at $o $ we find 
\begin{equation}
\label{betactions4} 
\beta_\xi \epsilon  |_o
= - \frac{c }{2} r \epsilon .
\end{equation} 
Since $r$ is invertible, if $\lambda =0 $ no SUSY is preserved. Suppose $\lambda \neq 0 $. We calculate
\begin{equation}
\label{dlambdact} 
\frac{1}{4} \mathrm{d} \xi ^\flat \cdot \epsilon 
=\frac{1}{4} \omega_{ ij } \mathrm{d} x^i \wedge \mathrm{d} x^j \cdot \epsilon 
=   \frac{1}{4} \omega_{ ij } \gamma^{ij } \epsilon .
\end{equation} 
Substituting $\omega_{ 12 } = \beta_1 $, $\omega_{ 34 } = \beta_2 $, we get
\begin{equation*}
 \frac{1}{4} \mathrm{d} \xi^\flat \cdot \epsilon  
 =  \frac{1}{2}  (\beta_1 \gamma^{ 12 } + \beta_2 \gamma^{ 34 } ) \epsilon .
\end{equation*} 
For convenience set
\begin{equation}
a =  \beta_1 \gamma^{ 12 } + \beta_2 \gamma^{ 34 },
\end{equation} 
then
\begin{equation}
L_\xi \epsilon 
=  \frac{1}{2} (  a \operatorname{Id} - c \, r  ) \epsilon .
\end{equation} 
Vanishing of the Lie derivative is thus  equivalent to the equations
\begin{equation}
\label{eqoiuerw}
( a - i c   )\epsilon_1=0 = (a + i c ) \epsilon_2 .
 \end{equation} 

One can check, e.g.~by using the explicit representation of the gamma matrices given in Appendix \ref{gmatr}, that  the eigenvalues of $a- ic $ are
\begin{equation}
\label{jklreopi} 
  i ( \pm \beta_1 \pm  \beta_2 -   c ),
\end{equation} 
with all four  possible sign combinations.  Thus
\begin{equation}
D =\det (a - i c )=\det (a +i c )=(\beta_1 + \beta_2 -c) (\beta_1 + \beta_2 + c )(\beta_1 - \beta_2 -c) (\beta_1 - \beta_2 + c ),
\end{equation}
and
\begin{equation} 
\label{condsusys4} 
D =0 \quad \Leftrightarrow \quad (\beta_1 \pm \beta_2 )^2 =  - |\varphi |^2 .
\end{equation} 
If $D =0 $  we see from (\ref{jklreopi}) that if $\beta_1 \beta_2  \neq 0 $  then $a  \pm  ic  $ has a 1-dimensional kernel, while if $\beta_1 \beta_2 =0 $  it has a 2-dimensional kernel. Correspondingly the fraction of preserved SUSY  is $\tfrac{1}{4} $ in the former case and $\tfrac{1}{2}$ in the latter. 
 \end{proof} 

\subsection{Geometry of the quotient}
\label{geomquots4} 
The quotient  $M  / \Gamma  $, where $\Gamma$ is the one-parameter group generated by the KVF $\xi =\partial_0 + \lambda $, is diffeomorphic to $S^4 $. Since $\xi$ is timelike, the induced metric is riemannian. To find an explicit expression for it we first rewrite the Minkowski space metric $\eta$ on $\mathbb{R} ^{ 1,5 } \supset - \mathbb{R}  \times S^4 $ in terms of more convenient coordinates.

Define $U = \exp ( - x^0 \lambda )\in SO(5 )$, so that 
\begin{equation}
\xi (f) =  U \partial_{ 0 } (U^{-1} f ).
\end{equation} 
 Let $y =(y^1 , \ldots  y^5 )$, $x =(x^1 , \ldots , x^5 )$, and set 
 \begin{equation}
 y = U x .
 \end{equation}
  Then  $\xi (y) =0 $, so $(y^1, \ldots , y^5  )$,   are good coordinates on the quotient. Since $\lambda \in \mathfrak{so }(5) $, its action on $\mathbb{R} ^5 $  is linear. Denote by $B$ the matrix representing the action of $\lambda$ with respect to the coordinates  $(x^1, \ldots , x^5 )$,  $\lambda x^i  = B^i_{ \ j } x^j $. Then
\begin{equation}
\mathrm{d} x =\mathrm{e}^{ x^0 B }  \left(  \mathrm{d} y + B y \, \mathrm{d} x^0 \right) ,
\end{equation} 
and
\begin{equation}
\begin{split}
\eta &
= - (\mathrm{d} x^0 )^2  + \mathrm{d} x^T \mathrm{d} x
=  - (\mathrm{d} x^0 )^2  +  (\mathrm{d} y + B y \mathrm{d} x^0 )^T  (\mathrm{e} ^{ x^0 B } )^T (\mathrm{e} ^{ x^0 B } ) (\mathrm{d} y +  B y \mathrm{d} x^0 )\\&
=- (1 - (B y )^T (B y ) ) \Big( \mathrm{d} x^0  - (1 - (B y )^T (B y ) )^{-1} (B y )^T  \mathrm{d} y  \Big)^2  \\ &
\phantom{ - }+  (1 - (B y )^T (B y ) )^{-1}\Big( (B y )^T  \mathrm{d} y  \Big)^2 + \mathrm{d} y^T \mathrm{d} y.
\end{split}
\end{equation} 
Setting
\begin{equation}
\Lambda =1 - (B y )^T  B y, \quad A = - \Lambda^{-1} (B y )^T \mathrm{d} y
\end{equation} 
we can write
\begin{equation}
\eta = \mathrm{d} y^T  \Big( 1 + \Lambda^{-1}  B y (B y )^T \Big) \mathrm{d} y  - \Lambda (\mathrm{d}  x^0 +  A )^2 .
\end{equation} 
Comparing with (\ref{kkquotmetric}) we see that  the metric on the quotient $S^4 $  is thus
\begin{equation}
\label{mmmmmetriccc} 
 h =\phi^\ast  k , \quad k = \mathrm{d} y^T  \left( 1 + \Lambda^{-1}  B y (B y )^T \right)  \mathrm{d} y ,
\end{equation} 
where $\phi $ is the standard embedding $S^4 \hookrightarrow \mathbb{R} ^{ 5 }$. Since $ y $ is related to $x$ by an orthogonal transformation,  restriction to $ S^4$ is achieved simply  by imposing $y^ T y = R^2  $.  As a sanity check note that $\lambda =0 \Rightarrow B =0 , \Lambda =1 $ so that the quotient is $S^4 $ with its round metric.

If  $\xi$  has the canonical form (\ref{kvfs4main}) then 
\begin{equation}
B =\left(\begin{array}{ccccc}0 &- \beta_1  & 0 & 0 & 0 \\\beta_1  & 0 & 0 & 0 & 0 \\0 & 0 & 0 & -\beta_2  & 0 \\0 & 0 &  \beta_2  & 0 & 0 \\0 & 0 & 0 & 0 & 0\end{array}\right)
\end{equation} 
and we get
\begin{equation}
\label{dsadasda} 
\begin{split} 
\Lambda & =1- \beta_1^2 (y_1^2 + y_2^2 )- \beta_2^2 (y_3^2 + y_4^2 ) ,\\
[( B y )^T \mathrm{d} y  ]^2  &= \beta_1^2 ( y^2 \mathrm{d} y^1 - y^1 \mathrm{d} y^2 )^2 + \beta_2^2 ( y^4 \mathrm{d} y^3 - y^3 \mathrm{d} y^4 )^2  .
\end{split} 
\end{equation} 
Note that $\Lambda >0 $ provided that $\beta_1 , \beta_2  $ satisfy  (\ref{sn23445}).  
Switching to double polar coordinates $y^1 =r \cos \alpha $, $y^2 =r \sin \alpha $, $y^3 =\rho \cos \gamma  $, $y^4 =\rho \sin \gamma  $ in (\ref{dsadasda}),  (\ref{mmmmmetriccc}) we get
\begin{equation}
\label{explmet} 
h =\phi^\ast k , \quad k = \mathrm{d} r^2 + \mathrm{d} \rho^2 + \mathrm{d} y_5^2 + \frac{(1- \beta_2^2 \rho^2 )r^2 \mathrm{d} \alpha^2  + (1- \beta_1^2 r^2 ) \rho^2 \mathrm{d} \gamma ^2 + 2\beta_1 \beta_2  r^2 \rho^2  \mathrm{d} \alpha \mathrm{d} \gamma }{1- \beta_1^2  r^2 - \beta_2^2 \rho^2 }.
\end{equation} 
The restriction $y^T  y =R^2 $  is now $y_5^2 + \rho^2 + r^2 =R^2 $.  The geometry described by  (\ref{explmet}) is that of a 2-torus fibration over 
the  interior of the round closed quarter 2-sphere
\begin{equation}
Q = \{ (y^5 , \rho , r ): r^2 + \rho^2 + y_5^2 =R^2, \, \rho \geq 0, \, r \geq 0 \}
\end{equation} 
attached to the boundary $\partial Q  $ where the circle fibres collapse to zero size.
 The fibres are parametrised by $(\alpha , \gamma  )$. Because of the constraint (\ref{sn23445}) the factors $1- \beta_1^2 r^2 $, $1- \beta_2^2 \rho^2 $,  $1- \beta_1 r^2 - \beta_2 \rho^2  $ never vanish, so the $\alpha$-circles have maximum radius $\tfrac{R}{\sqrt{ 1- \beta_1^2 R^2 }}$ for $r =R $ and collapse to zero size for $ r =0 $. The $\gamma $-circles have a similar behaviour with $r $ replaced by $\rho$. Independently translating and reflecting along the circle fibres gives an $O (2) ^2 $ subgroup of the isometry group. There is an additional $ \mathbb{Z} _2 $ isometry generated by $y^5 \mapsto - y^5 $ which fixes the base $Q$.
It is clear that the symmetry is enhanced if $\beta_2 =\pm \beta_1 $, in fact we see from Table \ref{s4quotiso} below that then $\mathfrak{l } = \mathfrak{so }(2) \oplus \mathfrak{so }(3)$.

If   $\beta_1 \beta_2 =0$  (\ref{mmmmmetriccc}) it is convenient to introduce polar coordinates in one plane only. Up to relabeling the coordinates, we can assume $\beta_2 =0 $ and set $\beta_1 =\beta $,  $y^1 = r \cos \alpha  $, $y^2 =r \sin \alpha  $. Then
\begin{equation}
\label{s4quotttt} 
\begin{split} 
k &
= g_{ E^3 } + \mathrm{d} r^2 + \frac{r^2 \mathrm{d} \alpha ^2 }{1 - \beta^2 r^2 } ,
\end{split} 
\end{equation} 
where 
\begin{equation} 
 g_{ E^3 } = (\mathrm{d} y^3)^2 +  (\mathrm{d} y^4 )^2 +( \mathrm{d} y^5 )^2
 \end{equation}
 is the Euclidean metric on $\mathbb{R} ^3 $. The geometry described by (\ref{s4quotttt}) is that of a $U (1) $ fibration over the interior of the round closed half 3-sphere 
 \begin{equation}
 S =\{ (y^3 , y^4 , y^5 , r ): y_3^2 + y_4^2 + y_5^2 + r^2 =R^2, \, r \geq 0 \} 
 \end{equation}
 attached to the boundary  $\partial S $  where  the circle fibres collapse to zero size. The angular coordinate $\alpha $ parametrises the circle fibres. The circle size is maximal for $r =R $ and decreases with $r$ reaching zero for $r =0$. It is clear from this description that the  isometry group of $\phi^\ast k $ is  $O (3) \times O (2) \subset O (5) $, with $O (3) $ arising as the isometry group of $S$. 
 
By (\ref{isoquotform})  the hereditary isometry algebra of   the metric (\ref{mmmmmetriccc}) on $S^4 $ is
\begin{equation}
\mathfrak{l } =\frac{ N_{X_\xi }(\mathbb{R}  \oplus \mathfrak{so} (5)    ) }{ \mathbb{R}   X_\xi  },
\end{equation}
where $X_\xi = (1 , - \omega  )\in \mathbb{R}  \oplus \mathfrak{so} (5)$ is the element  corresponding to the KVF $\xi$, $N_{X_\xi }(\mathbb{R}  \oplus \mathfrak{so} (5) )$ is the normaliser of  $X_\xi $ in  $\mathbb{R}  \oplus \mathfrak{so} (5) $ and $\mathbb{R}  X_\xi $ is the 1-dimensional Lie algebra generated by $X_\xi $. Denote by $N_\omega (\mathfrak{so }(5) )$ the normaliser of $\omega$ in $ \mathfrak{so }(5) $. Then $N_{ X_\xi } (\mathbb{R}  \oplus \mathfrak{so } (5) ) =\mathbb{R}  \oplus N_\omega  ( \mathfrak{so} (5)  )$ and we can use the quotient by $\mathbb{R}   X_\xi $ to fix the $\mathbb{R}  $ part to zero so that we only have to calculate $ N_ \omega  ( \mathfrak{so} (5)  )$.
Using (\ref{partcorrrlkj}) we find the result in Table \ref{s4quotiso}.
\begin{table}[h]
\caption{Hereditary isometry algebra $\mathfrak{l }$  of the KK reduction of $ - \mathbb{R}  \times S^4  $ by $\xi  =\partial_0 +  \beta_1 R_{ 12 }+ \beta_2 R_{ 34 } $.}
\centering
\begin{tabular}{|c|c|c|}
\hline 
 Conditions & $\mathfrak{l }$  & Generators\\ \hline
 $\beta_1 \beta_2 \neq 0, \, \beta_1 \neq \pm \beta_2  $& $\mathfrak{so }(2) \oplus \mathfrak{so }(2)$ & $\epsilon _{ 12 }; \epsilon _{ 34 }$ \\ \hline
 $\beta_1 \beta_2 \neq 0, \, \beta_1 = \pm \beta_2  $& $\mathfrak{so }(2) \oplus \mathfrak{so }(3)$ & $\epsilon _{ 12 } \pm  \epsilon _{ 34 }; \epsilon _{ 12 } \mp  \epsilon _{ 34 }, \epsilon _{ 13 } \pm  \epsilon _{ 24 }, \epsilon _{ 14 } \mp  \epsilon _{ 23 }$ \\ \hline
 $ \beta_2 = 0  $& $\mathfrak{so }(2) \oplus \mathfrak{so }(3)$ & $\epsilon _{ 12 }; \epsilon _{ 34 }, \epsilon _{ 35 }, \epsilon _{ 45 }$ \\ \hline
\end{tabular}
\label{s4quotiso}
\end{table}

 \section{Kaluza-Klein quotients of $\mathbb{R}  \times \mathrm{AdS}_{4 } $} 
 \label{ads4} 

 Recall that we take global coordinates $(y , x^1 , x^2 , x^3 , x^4 , x^5 ) \in \mathbb{R} \times \mathbb{R}^{ 2,3 } $, with $x^1 , x^2 $ timelike and the other coordinates spacelike. We consider $\mathrm{AdS}_4 $ as the universal cover of the quadric in $\mathbb{R}^{ 2,3 }$
 \begin{equation}
 x_1^2 +  x_2^2 - x_3^2 - x_4^2  - x_{5}^2 = R^2 .
\end{equation} 
The parallel vector field $\varphi$ is
\begin{equation}
\varphi = c \partial_y 
\end{equation} 
with
 \begin{equation}
 \label{rfhiads} 
R  |\varphi | = R c =\sqrt{ 2}.
 \end{equation} 

 \subsection{One parameters subgroups of $G$}
 \label{oneparads4} 
 The isometry group of $M =\mathbb{R}  \times \mathrm{AdS}_4 $  is the  direct product $\mathbb{R} \times O (2,3 )$ with Lie algebra   $\mathfrak{g}  = \mathbb{R}  \oplus  \mathfrak{so }(2,3)$. Let $ (\tau^y , - \omega )\in \mathfrak{g}  $. A computation similar to (\ref{xitoxxxxx}) gives the corresponding KVF,
\begin{equation}
\xi =\tau + \lambda , \quad   \tau = \tau^y \partial_y,\quad  
 \lambda = \tfrac{1}{2} \omega ^{ ij } (x_i \partial _j - x_j \partial_i ) .
\end{equation}
Therefore we have the correspondence
\begin{equation}
(\tau ^y , - \omega ) \in \mathfrak{g}  \quad  \longleftrightarrow \quad   \tau^y \partial_y + \frac{1}{2} \omega^{ ij  } R_{ i j } \in \mathfrak{X}  (M) ,
\end{equation} 
where, for $ i, j =1 , \ldots , 5$,
\begin{equation}
R_{ ij  }= x_i  \partial_j  - x_j  \partial_i  .
\end{equation}
In particular, taking generators  $ (\epsilon_{ ij })$ of $ \mathfrak{so }(2,3) $,  $ \epsilon_{ ji }=- \epsilon_{ ij }$,
\begin{equation}
\label{so23gen} 
\epsilon_{ ij } =
\begin{cases} 
 e_1 \otimes e^2 - e_2 \otimes e^1, \quad & \text{if }  (i,j )=(1,2),  \\
 - ( e_i \otimes e^j +  e_j \otimes e^i  ),  &\text{if }i=1,2,\,  j = 3,4,5, \\
e_j \otimes e^i  - e_i \otimes e^j ,  &\text{if }  3 \leq i<j \leq 5 , 
\end{cases} 
\end{equation} 
 with bracket
\begin{equation}
\label{so23bracket} 
[\epsilon_{ ij }, \epsilon_{ kl }]
=\eta^{ 2,3 }_{ ik } \epsilon_{ jl }+\eta^{ 2,3 }_{ jl } \epsilon_{ ik } - \eta^{ 2,3 }_{ il } \epsilon_{ jk } - \eta^{ 2,3 }_{ jk }\epsilon_{ il },
\end{equation} 
we get
\begin{equation}
\label{corrspecads} 
\epsilon_{ ij }\in \mathfrak{so }(2,3)  \quad \longleftrightarrow \quad R_{ ij } \in \mathfrak{X}  ( \mathrm{AdS}_4 ).
\end{equation}

\begin{proposition}\label{kvfads}
Let $\xi$ be a KVF on $M =  \mathbb{R}  \times  \mathrm{AdS}_4  $ and assume that $|\xi |$ never vanishes. Then there are coordinates such that, up to rescaling,
\begin{equation}
\label{xiadsxi} 
\xi =  \partial_y + \lambda ,
\end{equation} 
with $ \lambda\in \mathfrak{X}  (\mathrm{AdS}_4 )  $ one of the following KVFs:
\begin{align}
\label{sl1} 
\lambda_4 &=\beta R_{ 34 } , \quad \beta > 0,\\
 \label{sl2} 
\lambda_5  &=R _{ 13 } - R_{ 34 } ,\\
 \label{sl3} 
\lambda_6 &=R_{ 24 } + R_{ 34 } - R_{ 12 }- R_{ 13 },\\
 \label{sl4} 
\lambda_{ 11 }&=\beta  ( R_{ 13 } + R_{ 24 }) ,\quad \beta > 0,\\
 \label{t1} 
\lambda_1 &= \beta  R_{ 12 } , \quad  \beta >0, \quad  | \varphi | < \sqrt{ 2 }\beta  ,\\
 \label{t3} 
\lambda_8 &= R_{ 24 } - R_{ 13 }+ (1 + \beta )R_{ 34 } -(1 - \beta ) R_{ 12 }, \quad \beta <0, \quad  |\varphi | \leq - \sqrt{ 2 }\beta ,\\
\label{t2} 
\lambda_{ 10 }&= \beta_1   R_{ 12 } + \beta_2  R_{ 34 }, \quad  \beta_1> \beta_2 >0  , \quad |\varphi | < \sqrt{ 2 }\beta_1   ,\\
\label{t2bis} 
\lambda_{ 10* }&= \beta (  R_{ 12 } +   R_{ 34 }), \quad  \beta >0, \quad    | \varphi | < \sqrt{ 2 }\beta .
\end{align} 
The KVF $\xi =\partial_y + \lambda $  is spacelike for $\lambda$ given by (\ref{sl1})--(\ref{sl4}); timelike for $\lambda$ given by (\ref{t1})--(\ref{t2bis}).
\end{proposition} 

\begin{proof}
To avoid zeros the translation part $\tau$ of $\xi$ must be non-zero and we can rescale it so that $\tau = \partial_y $. Then $|\xi |^2 = 1  + |\lambda |^2 $ so we need either $|\lambda |^2 > -1 $, leading to a spacelike KVF, or $|\lambda |^2 < -1 $, leading to a timelike one.

The KVFs on $ \mathrm{AdS}_5 $, up to conjugation,  are listed in
\cite[Section~4.2.2]{Figueroa-OFarrill:2004lpm}.  We report the list here, following the same numbering.
\begin{equation*} 
\begin{alignedat}{2}
 \lambda& && |\lambda |^2 \\\hline
\lambda_1 &=\beta R_{ 12 },\,   \beta >0 & &- \beta^2 (x_1^2 + x_2^2 )\\
\lambda_2 &= \beta R_{ 13} ,\,  \beta >0 & & \beta^2 (x_1^2 - x_3^2 ),\\
\lambda_3 &=R_{ 12 }- R_{ 23 } &&- (x_1 + x_3 )^2 \\
\lambda_4 &=\beta R_{ 34 }, \,  \beta >0 &&\beta^2 (x_3^2 + x_4^2 ),\\
\lambda_5 &=R_{ 13 }- R_{ 34 } && (x_1 + x_4 )^2 ,\\
\lambda_6 &= R_{ 24 }+ R_{ 34 }- R_{ 12 }- R_{ 13 } & &0\\
\lambda_7 &= R_{ 24 }+ R_{ 34 }- R_{ 12 }- R_{ 13 } +  \beta (R_{ 14 } - R_{ 23 } ), \,  \beta >0\quad&&   4\beta (x_2 + x_3 )(x_1 + x_4 ) + \beta^2 ( R^2 + x_5^2 ),\\
\lambda_8 &= R_{ 24 }+ R_{ 34 }- R_{ 12 }- R_{ 13 } + \beta (R_{ 12 }+ R_{ 34 }), \,  \beta \neq 0&& 2\beta  \big[(x_1 + x_4 )^2 + (x_2 + x_3 )^2  \big]- \beta^2  (R^2 + x_5^2 ),\\
\lambda_9 &=\beta_1  (R_{12 }- R_{ 34 }) + \beta_2 (R_{ 14 }- R_{ 23 }),\,   \beta_i >0 && (\beta_2^2 - \beta_1^2)(R^2 + x_5^2 )-4 \beta_1 \beta_2 (x_1 x_3 + x_2 x_4 ) \\
\lambda_{10} &=\beta_1  R_{12 } + \beta_2  R_{ 34 },\,  \beta_i >0&&  \beta _2^2 (x_3^2 + x_4^2 )- \beta_1^2 (x_1^2 + x_2^2 ),\\
\lambda_{11} &=\beta_1  R_{13 } + \beta_2  R_{ 24 },\,   \beta_1 \geq \beta_2  >0&&  \beta_1^2  (x_1^2 - x_3^2 ) + \beta_2^2 (x_2^2 - x_4^2 ),	\\
\lambda_{12} &=\beta_1  R_{13 } + \beta_2  R_{ 45 },\,  \beta_i >0&&  \beta_1^2  (x_1^2 - x_3^2 ) + \beta_2^2 (x_4^2 + x_5^2 )	,\\
\lambda_{13} &=R_{ 12 } + R_{ 13 } + R_{ 15 } - R_{ 24 }- R_{ 34 }- R_{ 45 } && (x_1 - x_4 )^2 -4 (x_2 + x_3 )x_5 ,\\
\lambda_{14} &=R_{ 12 }- R_{ 23 } + \beta R_{ 45 }, \,  \beta >0&&  \beta^2 (x_4^2 + x_5^2 )- (x_1 + x_3 )^2 ,\\
\lambda_{15} &=R_{ 13 }- R_{ 34 } + \beta R_{25 }, \,  \beta >0 &&  (x_1 + x_4 )^2 + \beta^2 (x_2^2 - x_5^2 ).
\end{alignedat}
\end{equation*} 
The parameters  $\beta$, $\beta_i $  can always be chosen to satisfy the listed constraints  by conjugating within $SO(2,3) $. The norms of the vectors $\lambda_1 $--$ \lambda_{ 15 }$ can be computed by working on $\mathbb{R} ^{ 2,3 } $ and imposing the  $\mathrm{AdS}_5 $  constraint $ x_1^2 + x_2^2 - x_3^2 - x_4^2 - x_5^2 =R^2 $. We find the following bounds.
\begin{align*} 
\lambda_1 : \quad &- \infty  < |\lambda_1  |^2  \leq - \beta^2 R^2, \\
\lambda_2 : \quad &- \infty < |\lambda_2  |^2 < \infty ,\\
\lambda_3 : \quad &-  \infty < |\lambda_3  |^2 \leq 0, \\
\lambda_4 : \quad &0 \leq  |\lambda_4  |^2 < \infty ,\\
\lambda_5 : \quad &0 \leq  |\lambda_5 |^2 < \infty ,\\
\lambda_6 : \quad & | \lambda_6  |^2 =0,\\
\lambda_7 : \quad &- \infty < |\lambda_7  |^2 < \infty ,\\
\lambda_8 : \quad &\begin{cases} 
- \infty < |\lambda_8 |^2 < \infty & \beta >0,\\
- \infty < |\lambda_8  |^2 < - \beta^2 R^2  & \beta <0,\\
\end{cases} \\
\lambda_9 : \quad & - \infty < |\lambda_9  |^2 < \infty , \\
\lambda_{10} : \quad & \begin{cases} 
- \infty < |\lambda_{ 10} |^2 < \infty  & \beta_2^2 > \beta_1^2, \\
- \infty < |\lambda_{ 10}|^2 \leq - \beta_1^2 R^2   & \beta_2^2 \leq  \beta_1^2,
\end{cases} \\
\lambda_{11} : \quad & \begin{cases}
- \infty < |\lambda_{ 11} |^2 < \infty &\beta_1 > \beta_2 ,\\
\beta^2  R^2  \leq  |\lambda_{ 11}|^2 < \infty &\beta_1 = \beta_2 =\beta ,
\end{cases} 	\\
\lambda_{12} : \quad & - \infty < |\lambda_{12} |^2 < \infty ,\\
\lambda_{13} : \quad & - \infty < |\lambda_{ 13} |^2 < \infty,\\
\lambda_{14} : \quad & - \infty < |\lambda_{ 14} |^2 < \infty,\\
\lambda_{15} : \quad &- \infty < |\lambda_{ 15} |^2 < \infty .
\end{align*} 
In particular note that there are points  with arbitrarily small $| x_1 +x_4 |$, $| x_2+x_3| $ and arbitrarily large $R^ 2+x_ 5^2= (x_ 1+x_ 4)(x_1-x_	4)+(x_ 2+ x_3)(x_ 2-x_ 3) $. It follows that  $|\lambda_8 |^2 $ is unbounded below and, if $\beta <0 $, $|\lambda_8 |^2 <- \beta^2 R^2 $, while if $\beta >0 $ then $ |\lambda_8  |^2 $ is also unbounded above.
The vectors  $\lambda_4 $, $\lambda_5 $, $\lambda_6 $ satisfy $|\lambda_i  |^2 > -1 $. The vector  $\lambda_{ 11 } $ satisfies $|\lambda_{ 11 }|^2 > -1 $ provided that 
\begin{align*} 
\lambda_{ 11 }: \quad &\beta_1 =\beta_2 =\beta .
\end{align*} 
The vectors $\lambda_1 $, $\lambda_8 $, $\lambda_{ 10 }$ satisfy $|\lambda_i  |^2 < -1 $ provided that 
\begin{align*} 
\lambda_1 : \quad &\beta R >1,\\
\lambda_8 : \quad & \beta <0, \, -\beta R \geq 1,\\
\lambda_{ 10 }: \quad & \beta_1 \geq \beta_2 , \, \beta_1  R >1.
\end{align*} 
Rewriting the conditions  in terms of $|\varphi |$ using (\ref{rfhiads}) gives the stated result.
\end{proof} 
Note that since (\ref{xiadsxi}) always has a non-trivial translation part, the associated one-parameter group is non-compact and has the topology of a line.

\subsection{Preserved SUSY}

\begin{proposition}\label{psusyads4} 
The KVFs of Proposition \ref{kvfads} preserving a fraction $\nu>0$  of SUSY are given by Table \ref{presusyads}. 
\begin{table}[htp]
\caption{Fraction $\nu>0$  of SUSY preserved by the KVFs of Proposition \ref{kvfads}.}
\label{presusyads} 
\centering
\begin{tabular}{|c|c|c|}
\hline 
$\lambda $ & Condition & $\nu$  \\ \hline
$\lambda_4  $ &  $0< \beta = |\varphi | $ & $ 1/2$\\\hline
$\lambda_1 $ & $0<\beta = \sqrt{ 2 }|\varphi | $ &  $ 1/2$\\ \hline
$\lambda_6 $ &$ | \varphi |=2^{ 3/4 } $&$1/4$\\ \hline
$\lambda_8 $ &
$|\varphi |=\varphi_1 \text{ or }  |\varphi |=\varphi_2$ &$ 1/4$ \\  \hline
$\lambda_{ 10 } $ &
$( \beta_1 +  \sqrt{ 2 }\beta_2 =\sqrt{ 2 } |\varphi |, \beta _1 > \sqrt{ 2 }\beta _2 >0) \text{ or }  (|\beta_1 -  \sqrt{ 2 }\beta_2 | =\sqrt{ 2 } |\varphi |, \beta _1 > \beta _2 >0, \beta _1 \neq \sqrt{ 2 }\beta _2 )$
 &$1/4$\\  \hline
\hline
$\lambda_{ 10* } $ &$0< (\sqrt{ 2 } - 1 )\beta  =\sqrt{ 2 } |\varphi | $ & $ 1/4$\\ \hline
\end{tabular}
\end{table}

\end{proposition} 
For $\lambda=\lambda_8$ the values of $\varphi_1 $, $\varphi_2 $  are given by (\ref{pvs1})--(\ref{pvs2}) and the corresponding range of $\beta$ by (\ref{cvf1})--(\ref{cvf2}). The  conditions listed in Table \ref{presusyads} take into account both the constraints coming from Proposition \ref{kvfads} and those arising from SUSY preservation.
\begin{proof} 
$\varphi$ is the spacelike vector field
\begin{equation}
\varphi= c\partial_y 
\end{equation}
 with $c=| \varphi |  =\sqrt{ 2 }/ R $. Substituting in (\ref{xicontr}) we get
\begin{equation}
\label{betads} 
\beta_\xi \epsilon 
= \frac{|\varphi |}{4}  \xi ^i (\gamma_i \gamma_y + 3\eta_{ i y } ) r \epsilon 
= \frac{|\varphi |}{2}   \left( 1
+  \frac{1}{2}  \lambda  ^i \gamma_i \gamma_y \right) r \epsilon ,
\end{equation} 
so
\begin{equation}
\label{liederads} 
L_\xi \epsilon 
= \frac{1}{2}   \left[ |\varphi |\left( 1 +  \frac{1}{2}  \lambda  ^i \gamma_i \gamma_y \right) r 
+ \frac{1}{2} \omega_{ i j } \gamma^{ i j } \right]  \epsilon 
= \frac{1}{2}  \left[  br + a \right] \epsilon=0,
\end{equation} 
having defined
\begin{equation}
b =|\varphi |\left(  1 +  \frac{1}{2}  \lambda^i \gamma_i \gamma_y \right) ,
\quad  a =\frac{1}{2} \omega_{ i j } \gamma^{ i j }.
\end{equation} 

In order to compute $L_\xi \epsilon $ for $\xi$ one of the KVFs given in Proposition \ref{kvfads}  we  work at some convenient point $o$.
For $\lambda =\lambda_5 $ we take $o$ to have coordinates  $x^2  =R \Rightarrow x^1 =x^3 =x^4 =x^5 =0$ so that $( \partial_1 , \partial_y , \partial_3 , \partial_4 , \partial_5 ) $ is a local orthonormal frame for $\mathbb{R}  \times \mathrm{AdS}_4  $ at $o$. Since $ x^2 =R  $,  $\mathrm{d} x^2 |_o =0 $ hence  $\partial_2  $ acts by zero on spinors while $( \partial_1, \partial_y , \partial_3 , \partial_4 , \partial_5 )$ map to a representation  $(\gamma_1, \gamma_y , \gamma_3 , \gamma_4 , \gamma_5 )$ of $\Cl(1,4 )$ with $\gamma_1^2 =1 $, $\gamma_y^2 = \gamma_i^2 =-1 $, $i =3,4,5 $. 
For all the other KVFs we take  $o $ to have coordinates $x^1 =R \Rightarrow x^2 = x^3 =x^4 =x^5 =0 $ so that
$ (\partial_2,\partial_y , \partial_3, \partial_4, \partial_5)$ is an orthonormal frame for $\mathbb{R}  \times \mathrm{AdS}_4 $ at $o$. In this case $\partial_1 $ acts trivially on spinors while $( \partial_2 , \partial_y , \partial_3 , \partial_4 , \partial_5 )$ map to a representation  $(\gamma_2, \gamma_y , \gamma_3 , \gamma_4 , \gamma_5 )$ of $\Cl(1,4 )$ with $\gamma_2^2 =1 $, $\gamma_y^2 = \gamma_i^2 =-1 $, $i =3,4,5 $. An explicit choice of gamma matrices is given in Appendix \ref{gmatr}.
For $\lambda =\lambda_4 $ and $\lambda =\lambda_5 $ we have  $\lambda |_o =0 $   but in all the other cases $\lambda |_o \neq 0$.  With these choices, and substituting $R = \sqrt{ 2 }/ |\varphi |$, the values of $a|_o$, $b|_o $ are listed in Table \ref{xisusyads}.
\begin{table}[htp]
\caption{Values of $a |_o $, $b |_o $ for the KVFs of Proposition \ref{kvfads}.}
\centering
\begin{tabular}{|c|c|c|}
\hline 
$\lambda $ & $a |_o $ & $| \varphi |^{-1}  b |_o  $ \\ \hline
$\lambda_4  $ &  $\beta \gamma_{ 34 } $ &1 \\\hline
$\lambda_5  $ & $-\gamma_{ 13 } - \gamma_{ 34 } $& 1 \\ \hline
$\lambda_6  $ &$\gamma_{ 34 }-\gamma_{ 24 }  $ & $1- \tfrac{1}{\sqrt{ 2} |\varphi |} (\gamma_2 + \gamma_3 )\gamma_y  $\\ \hline
$\lambda_{ 11}$&$-\beta   \gamma_{ 24 } $& $1 + \tfrac{ \beta    }{\sqrt{ 2}|\varphi |} \gamma_3 \gamma_y $  \\\hline
$ \lambda_1  $ & $0$ &$1 +  \tfrac{\beta  }{\sqrt{ 2} |\varphi |}  \gamma_2 \gamma_y $\\ \hline
$\lambda_8 $ &  $ (1 + \beta ) \gamma_{ 34 } - \gamma_{ 24 } $ & $ 1 - \frac{1}{\sqrt{ 2} |\varphi |} (\gamma_3  + (1-\beta  )\gamma_2  ) \gamma_y $\\\hline
$ \lambda_{ 10 } $& $\beta_2 \gamma_{ 34 }$ & $1 + \frac{ \beta_1    }{\sqrt{ 2} |\varphi |} \gamma_2 \gamma_y $  \\ \hline
$ \lambda_{ 10* } $& $\beta  \gamma_{ 34 }$ & $1 + \frac{ \beta}{\sqrt{ 2} |\varphi |} \gamma_2 \gamma_y $  \\ \hline
\end{tabular}
\label{xisusyads}
\end{table}

Equation (\ref{liederads}) gives the two linear equations
\begin{equation}
(a + ib )\epsilon_1 =0 =(a -i b )\epsilon_2 .
\end{equation} 
It is clear that there are no non-trivial solutions if $\lambda =0 $. 
For  $\lambda =\lambda_5 $,  $a^2 =0 $ and $b $ is not nilpotent so  no SUSY is preserved. The case $\lambda =\lambda_{ 11 }$  also does not preserve any SUSY.
In the other cases computing the rank of the matrices making use of an explicit representation, such as the one given in Appendix \ref{gmatr}, gives the following result.
\begin{alignat*}{3}
& \lambda & & \text{Condition}\quad  & & \text{Fraction of SUSY}  \\ \hline
&\lambda_4  \quad  &&|\beta |= |\varphi |&& 1/2 \\
& \lambda_1 & &|\beta |  =  \sqrt{ 2 } | \varphi | \quad &&1/2 \\ 
&\lambda_{ 10 } \quad & & 
| \beta_1 \pm \sqrt{ 2 }\beta_2 | =\sqrt{ 2 } |\varphi |
\quad  &&1/4 \\ 
&\lambda_{ 10 *} &&
(\sqrt{ 2 }\pm 1 ) |\beta | =\sqrt{ 2 } |\varphi |
\quad  & &1/4 \\ 
&\lambda_6  && |\varphi |=2^{ 3/4 }  &&1/4  \\ 
&\lambda_8  \quad && \text{see (\ref{pvs1})--(\ref{pvs2})  }  &&1/4 
\end{alignat*}
For $\lambda =\lambda_8 $ some SUSY is preserved if and only if  $|\varphi | =\varphi_1 $ or $|\varphi |= \varphi_2 $, with
\begin{align}
\label{pvs1} 
\varphi_1 &= \sqrt{\beta + \frac{3}{2} \beta ^2 + \sqrt{ 2 } | \beta  ^2 -2 |},\\
\label{pvs2}
\varphi_2 &= \sqrt{\beta + \frac{3}{2} \beta ^2 - \sqrt{ 2 } | \beta  ^2 -2 |} .
\end{align} 
where the value of $\beta$ is constrained by the condition that $ \varphi _1  $, $\varphi _2 $  are real.

Combining the conditions listed above with those, listed in Proposition \ref{kvfads},  coming from imposing that the KK reduction results in a smooth manifold  gives the stated result.  In the case of $\lambda _8 $  the combined conditions read
\begin{align} 
\label{cvf1} 
|\varphi |&=\varphi _1: \quad  - \left( \frac{\sqrt{ 17-4 \sqrt{ 2 }} + 1}{2 \sqrt{ 2 } -1} \right) \leq \beta  \leq 
- \left( \frac{\sqrt{ 17+4 \sqrt{ 2 }} - 1}{2 \sqrt{ 2 } +1} \right) ,\\
\label{cvf2} 
|\varphi |&=\varphi _2 : \quad  
- \left( \frac{\sqrt{ 17-12 \sqrt{ 2 }} + 1}{3-2 \sqrt{ 2 } } \right) < \beta  < 
- \left( \frac{\sqrt{ 17+12 \sqrt{ 2 }} +  1}{3 + 2 \sqrt{ 2 } } \right).
\end{align} 
\end{proof} 

\subsection{Geometry of the quotient}
The quotient $( \mathbb{R}  \times \mathrm{AdS}_4 )/ \Gamma_\xi $, for $\xi$ one of the KVFs listed in Proposition \ref{kvfads},  is diffeomorphic to $\mathbb{R}  ^4 $ equipped with a lorentzian or riemannian metric depending on the causal character of $\xi$.
In order to find the quotient metric we proceed similarly as we did in Section \ref{geomquots4},  first working on $\mathbb{R}  \times \mathbb{R}^{ 2, 3 } $ and then restricting to $\mathbb{R}  \times \mathrm{AdS}_4 $.

Let
\begin{equation}
\eta^{ 2,4 }=  \mathrm{d} z^2 - \mathrm{d} x_1^2 - \mathrm{d} x_2^2 + \mathrm{d} x_3^2 + \mathrm{d} x_4^2 + \mathrm{d} x_5^2  
\end{equation} 
be the flat metric on $\mathbb{R}  \times \mathbb{R} ^{ 2,3 }$, $ x =(x^1 , x^2 , \ldots , x^5   )\in \mathbb{R} ^{ 2,3 }$, $z$ a global coordinate on the $\mathbb{R}  $ factor. Define $M^\dagger =M^T \eta^{ 2,3 }$.  The KVF $\xi =\partial_z + \lambda $ can be written
\begin{equation}
\xi =U  \partial_z U^{-1}  
\end{equation}
 with $U =\exp(-z \lambda ) $. The coordinates $(y^i )$ defined by  $y = U x $ are good coordinates on the orbit space. The action of $\mathfrak{so }(2,3 )$ on $\mathbb{R} ^{ 2,3 }$ is linear so let $B$ be the matrix representing $\lambda$ with respect to the $x$  coordinates,
$\lambda x^i  = B^i_{ \ j } x^j $. Then
\begin{equation*}
\begin{split}
\eta^{ 2,4 }&
=\Lambda (\mathrm{d} z + A )^2 +  \mathrm{d} y^\dagger \left( 1 - \Lambda^{-1}B y\,  (B y )^\dagger  \right) \mathrm{d} y,
\end{split}
\end{equation*} 
with
\begin{equation}
\Lambda =1 + (B y )^\dagger B y, \quad 
A = \Lambda ^{-1}  (B y )^\dagger \mathrm{d} y.
\end{equation} 
The quotient metric is 
\begin{equation}
 h =\phi^\ast k , \quad k=\mathrm{d} y^\dagger \left( 1 - \Lambda^{-1}B y\,  (B y )^\dagger  \right) \mathrm{d} y
\end{equation} 
where $\phi^\ast $ is the restriction to $\mathrm{AdS}_4 \subset \mathbb{R} ^{ 2,3 }$, which is achieved by imposing $y^T \eta^{ 2,3 } y = x^T \eta  ^{ 2,3 }x =- R^2 $.

Let us consider some cases in more detail. If $\lambda = R_{ k l }$ then $B$ has components
\begin{equation}
B^i_{ \ j } =\eta^{ 2,3 }_{ k j }\delta^i_l - \eta^{ 2,3 }_{ l j }\delta^i_k. 
\end{equation} 
For $ \lambda =\lambda_{ 10 } $, the quotient is a riemannian manifold which is the hyperbolic equivalent of (\ref{explmet}). Introducing double polar coordinates $y^1 =r \cos \alpha $, $y^2 =r \sin \alpha $, $y^3 =\rho \cos \gamma  $, $y^4 =\rho \sin \gamma  $ we get
\begin{equation}
\label{explmetads} 
 k = - \mathrm{d} r^2 + \mathrm{d} \rho^2 + \mathrm{d} y_5^2 
+ \frac{ (1 +  \beta_2^2 \rho^2 )r^2 \mathrm{d} \alpha^2  +  (\beta_1^2 r^2 - 1 ) \rho^2 \mathrm{d} \gamma ^2 
+ 2\beta_1 \beta_2  r^2 \rho^2  \mathrm{d} \alpha \, \mathrm{d} \gamma }{ -1 +  \beta_1^2  r^2 -  \beta_2^2 \rho^2 }.
\end{equation} 
Because of the constraints $ \beta_1   > \beta_2  $, $\beta_1 R > 1 $, we have $ \beta_1 r^2 > 1 $, $-1 +  \beta_1^2  r^2 -  \beta_2^2 \rho^2>0  $.
The geometry described by (\ref{explmetads})   is that of a  2-torus fibration over the interior of $Q$, where $Q$  is the portion of hyperbolic 2-space given by
\begin{equation}
Q =\{ (r, \rho ,  y^5) \in \mathbb{R} ^3 : - r^2 + \rho^2  + y_5^2 =- R^2 , \,  \rho  \geq 0 \} ,
\end{equation} 
collapsing to a circle fibration on the boundary $\partial Q $.
 The circles are parametrised by $( \alpha , \gamma )$ and while the  $\alpha$ fibres always have non-zero length, the $\gamma$ fibres collapse on the boundary $ \rho =0 $ of $Q$.

If $\beta _1 =0 $, so that $\lambda =\lambda_4 $, it is convenient to introduce polar coordinates  in the $( y_3 , y_4 ) $-plane only. The KVF is now spacelike so we get the lorentzian metric
\begin{equation}\label{eq:new-lor-4d-metric-ads}
k =	 - \mathrm{d} y_1^2 - \mathrm{d} y_2^2  + \mathrm{d} y_5^2 + \mathrm{d} r^2  + \frac{r^2 \mathrm{d} \alpha^2}{1 + \beta^2 r^2 }
\end{equation} 
which describes a circle bundle over the interior of $S$, for $S$ ``half'' $\mathrm{AdS}_3 $, 
\begin{equation}
S = \{ (y^1 , y^2 , y^5, r ): - y_1^2 - y_2^2 + y_5^2 + r^2 =- R^2 , \, r \geq 0 \} ,
\end{equation} 
with the circle fibres always of finite length and  collapsing to zero on the boundary $r =0 $. The isometry group of $S$ is $O (2,1) $.

If $\beta _2 =0 $, so that $\lambda = \lambda_1 $,	 it is convenient to introduce polar coordinates  in the $( y_1 , y_2 ) $-plane only. The KVF is  timelike and we get the riemannian  metric
\begin{equation}\label{eq:new-riem-4d-metric-ads}
k =	 - \mathrm{d}r^2 + \mathrm{d} y_3^2  + \mathrm{d} y_4^2 + \mathrm{d} y_5 ^2  + \frac{r^2 \mathrm{d} \alpha^2}{  \beta^2 r^2-1 }
\end{equation} 
which describes a circle bundle over the interior of $S$, for $S$ ``half'' hyperbolic 3-space
\begin{equation}
S = \{ (y^3 , y^4 , y^5, r ): y_3^2 +  y_4^2 +  y_5^2 -r^2 = - R^2 , \, r \geq 0 \} ,
\end{equation} 
with the circle fibres always of finite length and  collapsing to zero on the boundary $r =0 $. The isometry group of $S$ is $O (3) $.

The hereditary isometry algebra is
\begin{equation}  
\mathfrak{l }=\frac{ N_{X_\xi }(\mathbb{R}  \times \mathfrak{so }(2,3 )) }{ \mathbb{R}   X_\xi  },
\end{equation} 
where $X_\xi $ is the Lie algebra element corresponding to the vector field $\xi $ and  $ N_{X_\xi }(\mathbb{R}  \times \mathfrak{so }(2,3 ))  $ its normaliser in $ \mathbb{R}  \times \mathfrak{so }(2,3 )$.  Taking the generators (\ref{so23gen}) and using (\ref{corrspecads}) 
 one finds the result given in Table \ref{adsgeneratorslie}.
 
\begin{table}[h]
\caption{Hereditary isometry algebra $\mathfrak{l }$  of the KK reductions of $  \mathbb{R}  \times \mathrm{AdS} ^4  $ by the KVFs of Proposition \ref{kvfads}.}
\centering
\begin{tabular}{|c|c|c|}\hline
$\lambda  $ & $\mathfrak{l } $ & generators \\ \hline
$\lambda_4 $  & $\mathfrak{so } (2 ) \oplus \mathfrak{so }(2,1 ) $& $\epsilon _{ 34 };\epsilon _{ 15 }, \epsilon _{ 25 }, \epsilon _{ 12 }$ \\\hline
$\lambda_5   $& $\mathbb{R} ^2  \ltimes \mathbb{R} ^3  $&$\epsilon _{ 14 };\epsilon _{ 25 } ; \epsilon _{ 34 }- \epsilon _{ 13 }, \epsilon _{ 24 }- \epsilon _{ 12 };\epsilon _{ 54 } -\epsilon _{ 15 }$ \\ \hline
$\lambda_6  $&$\mathfrak{co }(2,1) \ltimes \mathfrak{h}$&
$ \epsilon _{ 14 }- \epsilon _{ 23 }, \epsilon _{ 13 } + \epsilon _{ 24 } , \epsilon _{ 12} + \epsilon _{ 34 }, \epsilon _{ 14 }+ \epsilon _{ 23 };$\\ &&$\epsilon _{ 25 }+ \epsilon _{ 35 }, \epsilon _{ 15 }+ \epsilon _{ 45 }, X_{\lambda_6}$ \\ \hline
$ \lambda_{ 11 }$& $\mathfrak{so }(2) \oplus \mathfrak{so }(2,1)$&$\epsilon _{ 13 } + \epsilon _{ 24 } ; \epsilon _{ 12 } - \epsilon _{ 34 }, \epsilon _{ 13 } - \epsilon _{ 24 },\epsilon _{ 14 } + \epsilon _{ 23 }$\\\hline
$\lambda_1 $&$\mathfrak{so }(2 ) \oplus \mathfrak{so }(3 ) $&$ \epsilon _{ 12 }; \epsilon _{ 34 }, \epsilon _{ 35 }, \epsilon _{ 45 } $ \\\hline
$\lambda_8 $ & $\mathfrak{so }(2) \oplus \mathfrak{so }(2)  $& $ \epsilon _{ 12 }+ \epsilon _{ 34 }; \epsilon _{ 13 }- \epsilon _{ 24 }-2 \epsilon _{ 34 }$ \\\hline
$\lambda_{ 10 }  $& $\mathfrak{so }(2 ) \oplus \mathfrak{so }(2 )$& $\epsilon _{ 12 }; \epsilon _{ 34 }$\\\hline
$\lambda_{ 10*}$& $\mathfrak{so }(2 ) \oplus \mathfrak{so }(2,1)$& $\epsilon _{ 12 } +    \epsilon _{ 34 };\epsilon _{ 12 }-  \epsilon _{ 34 }, \epsilon _{ 13 } -    \epsilon _{ 24 }, \epsilon _{ 14 }+    \epsilon _{ 23 } $\\ \hline
\end{tabular}
\label{adsgeneratorslie}
\end{table}
The  cases of $\lambda_5 $ and $\lambda_6 $ warrant some additional discussion.

If $\lambda = \lambda_5 $  define 
\begin{equation}
N_{ i } = \epsilon _{ i 4 }- \epsilon _{ 1i } , \quad i =2,3,5 .
\end{equation}
 Then
\begin{equation} 
\begin{split} 
[ N_{ i} , N_{ j }]&= [ \epsilon _{ 14 }, \epsilon_{ 25 }] = [\epsilon_{ 25 },N_{ 3 } ]=0,\\
[\epsilon_{ 14 }, N_{ i  }] &= -N_{ i },\\
[\epsilon_{ 25 }, N_{ 2 }] &= -N_{ 5 },\\
[\epsilon_{ 25 }, N_{ 5 }] &= -N_{ 2 }.
\end{split} 
\end{equation} 
Therefore the Lie algebra has the structure  
\begin{equation}
\mathfrak{l }=\mathbb{R} ^2 
\ltimes   \mathbb{R}^3 ,
\end{equation} 
with $\mathbb{R} ^2 =\operatorname{Span}_\mathbb{R}  (\epsilon_{ 14 }, \epsilon_{ 25 } )$, $\mathbb{R}^3    =\operatorname{Span}_ \mathbb{R}  ( N_{ 2 },  N_3 , N_{ 5 } )$.

If $\lambda = \lambda_6 $ then $\mathfrak{l }$ is generated by e.g.~$\epsilon _{ 14 }, \epsilon _{ 23 }, \epsilon _{ 13 }- \epsilon _{ 34 }, \epsilon _{ 13 } + \epsilon _{ 24 } , \epsilon _{ 34 } + \epsilon _{ 12 }$, $\epsilon _{ 15 }+ \epsilon _{ 45 } $ and $ \epsilon _{ 25 }+ \epsilon _{ 35 }$. Defining
\begin{equation}
\begin{split} 
x_1& =\epsilon _{ 14 } - \epsilon _{ 23 }, \quad 
x_2 = \epsilon _{ 13 } + \epsilon _{ 24 },\quad 
x_3  = - ( \epsilon _{ 12 }+ \epsilon _{ 34 }),\\
y_1 &=\epsilon _{ 15 } + \epsilon _{ 45 }, \quad 
y_2 =\epsilon _{ 25 }+ \epsilon _{ 35 }, \quad a = -(\epsilon _{ 14 }+ \epsilon _{ 23 }), \quad b  = -X_{\lambda_6},
\end{split} 
\end{equation} 
we find
\begin{equation}
\begin{split}
[x_1 ,x_2 ] &= - 2 x_3 , \quad [ x_2 , x_3 ] = 2 x_1 , \quad [ x_1 , x_3 ] =-2 x_2 ,\\
[ y_1 , y_2 ] &= b, \quad [b, y_1 ] =[b,y_2 ] =0,\\
[x_1 , y_1 ] &= - y_1 , \quad  [ x_1 , y_2 ] = y_2 , \quad [x_2 , y_1 ]  =- y_2 , \quad [ x_2 , y_2 ] = - y_1, \\
 [x_3 , y_1 ] &=y_2 , \quad [x_3 , y_2 ] = - y_1, \quad  [ x_i,b ]=0, \\
[a, x_i ]&=0, \quad [a, y_1 ] = y_1 , \quad [a , y_2 ] =y_2 , \quad [a  , b ]= 2 b .
\end{split}
\end{equation} 
Hence  $\mathfrak{l} =\mathfrak{co }(2,1) \ltimes \mathfrak{h}  $  with $\mathfrak{h}  = \operatorname{Span}  ( y_1 , y_2 , b ) $  the 3-dimensional Heisenberg algebra, and $ \{ x_1, x_2 , x_3 ,a \} $ spanning the Lie algebra $\mathfrak{c o }(2,1)= \mathbb{R}  \oplus \mathfrak{so }(2,1 )$ of isometries and dilations of $\mathbb{R}^{ 1,2 }$.\footnote{That is the Lie algebra of the group $\{ A \in GL(3, \mathbb{R}): A^T \eta^{ 2,1 } A =c \eta^{ 2,1 }, \, c \in \mathbb{R}^\times  \}$.}
 Note that $\mathfrak{l }$ is a graded algebra with  $\{ x_i, a \} $ in degree $0$, $ \{ y_i \} $ in degree $1$, $b$ in degree 2 and $a$ acting as a grading element.

\section{Kaluza-Klein quotients of the Cahen-Wallach background}
\label{cw}

Some details on CW spaces are given in Appendix \ref{cwgener}. We recall that a 5-dimensional CW space has the topology of  $\mathbb{R} ^5 $ with global coordinates $(x^+ , x^- , x^i) $, $i =1,2,3 $, and the lorentzian metric
\begin{equation}
\label{cwmmmk} 
g =2 \mathrm{d} x^+ \mathrm{d} x^-  + \sum_{ i , j =1 }^3 A_{ ij } x^i x^j   (\mathrm{d} x^- )^2  +  | \mathrm{d} x |^2 .
\end{equation} 
Here $|\mathrm{d} x |^2 = \mathrm{d} x_1^2 + \mathrm{d} x_2^2 + \mathrm{d} x_3^2 $ and $A$ is a symmetric bilinear form on $\mathbb{R} ^3 $ which in our case is simply the Euclidean inner product,
\begin{equation} 
\label{ouraaa} 
 A = \eta^{0,3 }, 
 \end{equation}
so that
 \begin{equation}
 \label{cw5metric} 
g =2 \mathrm{d} x^+ \mathrm{d} x^-  +  |x |^2    (\mathrm{d} x^- )^2  +  | \mathrm{d} x |^2 .
\end{equation} 
The parallel vector field $\varphi$ is
\begin{equation}
\varphi =\sqrt{ 2 }\, \partial _+.
\end{equation}

\subsection{One parameters subgroups of $G$}
\label{oneparcw} 
As discussed in Appendix \ref{cwgener},  the isometry algebra of a CW space with symmetric bilinear form $A= \eta^{ 0,3 }$ is
\begin{equation} 
\label{hhhhyuti} 
\mathfrak{g}  \rtimes \mathfrak{so} (3) 
\end{equation} 
where  $\mathfrak{so } (3) $ is the Lie algebra of the isometry group of $\eta^{ 0,3 }$, whose generators $ (V_i )$ satisfy
\begin{equation}
[ V_i, V_j ] =-\epsilon_{ ijk }V_k ,
\end{equation} 
and $\mathfrak{g}  $ is the 8-dimensional Lie algebra with generators $ (e_i , e_i^\ast , e_+ , e_- )$, $i =1,2,3 $, and non-trivial brackets
\begin{equation}
\label{cwal8} 
[ e_- , e_i ] = e_i^\ast , \quad [ e_- , e_i^\ast ]=  e_i , \quad [ e_i^\ast , e_j ] =  \delta _{ ij }e_+ .
\end{equation} 
The action of $\mathfrak{so} (3) $ on $\mathfrak{g}$  in (\ref{hhhhyuti})   is the natural action of $\mathfrak{so }(3) $ on $\mathbb{R} ^3 $ on  $\operatorname{Span} (e_i )  $,   the adjoint action on $ \operatorname{Span} (e_i^\ast ) $, and the trivial action on $ \operatorname{Span} (e_+ , e_- )$,
\begin{equation}
[ V_i , e_j ] =  V_i e_j =- \epsilon_{ ijk }e_k, \quad [ V_i , e^\ast_j ] =  V_i e^\ast _j = -\epsilon_{ ijk }e^\ast_k, \quad [ V_i , e_\pm ]=0.
\end{equation}

Note that $(e_i , e_i^\ast , e_+ )$ form a representation of the 7-dimensional Heisenberg algebra. An explicit matrix representation of (\ref{cwal8}) is given by
\begin{equation}
\label{rephe8} 
\begin{split} 
e_1& =\begin{pmatrix}
0 &1 &0 &0 & 0 \\
0 &0 &0 &0 &-1\\
0 &0 &0 &0 & 0 \\
0 &0 &0 &0 & 0 \\
0 &0 &0 &0 & 0
\end{pmatrix} , \quad 
e_2 =\begin{pmatrix}
0 &0 &1 &0 & 0 \\
0 &0 &0 &0 &0\\
0 &0 &0 &0 & -1 \\
0 &0 &0 &0 & 0 \\
0 &0 &0 &0 & 0
\end{pmatrix} ,\quad 
e_3 =\begin{pmatrix}
0 &0&0 &1 & 0 \\
0 &0 &0 &0 &0\\
0 &0 &0 &0 & 0 \\
0 &0 &0 &0 &-1 \\
0 &0 &0 &0 & 0
\end{pmatrix} ,\\
e_1^\ast  &=\begin{pmatrix}
0 &1 &0 &0 & 0 \\
0 &0 &0 &0 &1\\
0 &0 &0 &0 & 0 \\
0 &0 &0 &0 & 0 \\
0 &0 &0 &0 & 0
\end{pmatrix} , \quad 
e_2^\ast  =\begin{pmatrix}
0 &0 &1 &0 & 0 \\
0 &0 &0 &0 &0\\
0 &0 &0 &0 & 1 \\
0 &0 &0 &0 & 0 \\
0 &0 &0 &0 & 0
\end{pmatrix} , \quad 
e_3^\ast  =\begin{pmatrix}
0 &0 &0 &1 & 0 \\
0 &0 &0 &0 &0\\
0 &0 &0 &0 &0 \\
0 &0 &0 &0 & 1 \\
0 &0 &0 &0 & 0
\end{pmatrix} , \quad \\
e_-  &=\begin{pmatrix}
1 &0 &0 &0 & 0 \\
0 &0 &0 &0 &0\\
0 &0 &0 &0 &0 \\
0 &0 &0 &0 & 0 \\
0 &0 &0 &0 & 1
\end{pmatrix} , \quad
e_+ =\begin{pmatrix}
0 &0 &0 &0 &- 2 \\
0 &0 &0 &0 &0\\
0 &0 &0 &0 &0 \\
0 &0 &0 &0 & 0 \\
0 &0 &0 &0 & 0
\end{pmatrix},\\
V_1 &= \begin{pmatrix}
0 &0 &0 &0 &0 \\
0 &0 &0 &0 &0\\
0 &0 &0 &1 &0 \\
0 &0 &-1 &0 & 0 \\
0 &0 &0 &0 & 0
\end{pmatrix}, \quad
V_2 = \begin{pmatrix}
0 &0 &0 &0 &0 \\
0 &0 &0 &-1 &0\\
0 &0 &0 &0 &0 \\
0 &1 &0 &0 & 0 \\
0 &0 &0 &0 & 0
\end{pmatrix}, \quad
V_3 = \begin{pmatrix}
0 &0 &0 &0 &0 \\
0 &0 &1 &0 &0\\
0 &-1 &0 &0 &0 \\
0 &0 &0 &0 & 0 \\
0 &0 &0 &0 & 0
\end{pmatrix}.
\end{split} 
\end{equation} 

As shown in Appendix \ref{cwkvfasdlk}, in terms of the coordinates $(x^i , x^+ , x^- )$ used in (\ref{cwmmmk}),  the KVFs corresponding to $(e_i , e_i^\ast , e_+ , e_- ) $ are
\begin{align} 
\label{xiepm} 
\xi_{ e_ \pm } &= \partial_{ \pm },\\
\label{xiei} 
\xi_{ e_i }&= \cosh (x^- ) \partial_i -  \sinh (x^- )  x^i  \partial_+  ,\\
\label{xieis} 
\xi_{ e_i^\ast } &
= \cosh (x^- ) x^i  \partial_+ -  \sinh (x^- ) \partial_i .
\end{align} 
We choose the generators $(V_1 , V_2 , V_3 )$ of $\mathfrak{so }(3) $ so that the associated KVFs are the usual  generators of rotations in $\mathbb{R} ^{ 0,3 }$,
\begin{equation}
\label{xivi} 
\xi_{ V_1 } = R_{ 23 }, \quad \xi_{ V_2 } = R_{ 31 }, \quad \xi_{ V_3 } = R_{ 12 }.
\end{equation}

\begin{proposition} 
\label{kvfcw} 
Let $\xi$ be a KVF of $M = \mathrm{CW}_5  $ with the metric (\ref{cw5metric})  and assume that $|\xi | $ never vanishes. Then there are coordinates such that, up to rescaling, $\xi $ is the KVF  associated to one of the following elements  $X_i \in \mathfrak{g}\rtimes \mathfrak{so }(3) $,
\begin{align}
\label{1case1}  X_1  &= e_- + b  V_3  +  \gamma e_+ , \quad   \gamma >0, \\
\label{2case2}  X_2 &=  V_3 + ce_3   , \, c\neq 0,\\
\label{4case4} X_4^\pm  & = V_3 + c (e_3 \pm e_3^\ast ), \quad  c \neq 0,\\
\label{6case6} X_6 &= e_3, \\
\label{8case8} X_8^\pm &= e_3 \pm e_3^\ast ,\\
\label{9case9} X_9  &= e_3 + d_3 e_3^\ast + d_1 e_1^\ast , \quad d_1 \neq 0.
\end{align} 
The KVFs corresponding to (\ref{1case1})--(\ref{9case9}) are all spacelike.
\end{proposition} 
\begin{proof}
Let $X $ be a generic element of $  \mathfrak{g}   \rtimes \mathfrak{so }(3) $,
\begin{equation}
X = \alpha e_- + b_i V_i + c_i e_i + d_i e_i^\ast + \gamma e_+ .
\end{equation} 
We now act by conjugation using the equations given in Section \ref{conjac} to simplify the form of $X$ as much as possible.

Suppose first $\alpha \neq 0 $ and rescale so that $\alpha =1 $. Conjugating first by $ \exp ( x_k e_k^\ast  )$ and then by $ \exp (y_k e_k )$  brings $X$ to the form
\begin{equation}
X \mapsto e_- + b_i V_i + (c + b \times y -x )_i e_i + (d + b \times x -y )_i e_i^\ast + \tilde \gamma e_+ 
\end{equation} 
where we do not need the explicit expression of $\tilde \gamma $. Imposing the coefficients of $e_i $ and $e_i^\ast $ to vanish we thus get the vectorial equations
\begin{align}
x &=c + b \times  y,\\
y &= d+ b \times x,
\end{align} 
which are solved by taking
\begin{equation}
\begin{split}
x &= \frac{c+ b \times d + (b \cdot c )b}{1 + |b |^2 }, \qquad 
y = \frac{d+ b \times c + (b \cdot d )b}{1 + |b |^2 }.
\end{split} 
\end{equation} 
Rotating (conjugation by $V_i $) so to align $V$ to the third direction and relabelling  the coefficients we get
\begin{equation}
\label{case1} 
X \mapsto e_- + b V_3 + \gamma e_+ .
\end{equation} 

Suppose now $\alpha =0 $, $b \neq 0 $. Then the same procedure as before gives 
\begin{equation*}
X \mapsto  b_i V_i + (c + b \times y  )_i e_i + (d + b \times x  )_i e_i^\ast + \tilde \gamma e_+ .
\end{equation*} 
In this case the equation
\begin{equation*}
c + b \times y =0
\end{equation*} 
has solution if and only if $b \cdot c =0 $, in which case $y = \tfrac{b \times c}{|b |^2 } $. Similarly $ d+ x \times b =0 $ has solution if and only if $b \cdot d =0 $, in which case $x = \tfrac{b \times d}{|b |^2 } $. Therefore we can only kill the part of $c_i e_i $ or $d_ie_i^\ast $ which is normal to $b_i V_i $. Unless $c =d=0 $ we can still conjugate by $x_k e_k $ or $x_k e_k^\ast $ with $x$ such that $x \times  b =0 $ so to kill the $e_+ $ part while not affecting the $e_i $, $e_i^\ast $ part. Rescaling and rotating we thus get
\begin{equation*}
X \mapsto V_3 +  c e_3 + de_3^\ast .
\end{equation*} 
We now conjugate by $\exp( x^- e_- )$ obtaining
\begin{equation}
\label{conjacasy}
X \mapsto V_3 +( c \cosh x^-  +d \sinh x^-) e_3 + (c\sinh x^- + d\cosh x^-) e_3^\ast .
\end{equation} 
Since $\tanh: \mathbb{R}  \rightarrow (-1,1 )$ is surjective, if $ |c/d| <1 $ we can kill the coefficient of $e_3 $ and if $ |d/c|<1 $ we can kill that of $e_3^\ast $. Therefore if $ \alpha =0 $, $b \neq 0 $ then  $X$ can be brought to one of the following forms
\begin{equation}
X = V_3 + \gamma e_+ , \quad X = V_3 + ce_3 , \quad X =V_3 + de_3^\ast , \quad X = V_3 + c ( e_3 \pm e_3^\ast ).
\end{equation} 

If $\alpha =0 =b $ but $c,d$ are not both zero we can eliminate the $e_+ $ part. If $c =0 $ then we can rotate and rescale so to get
\begin{equation}
\label{simpe3} 
X =e_3^\ast.
\end{equation} 
If $c \neq 0 $ rotating and rescaling we obtain
\begin{equation}
\label{cdneq0} 
X = e_3 +  d_3 e_3^\ast + d_1 e_1^\ast.
\end{equation} 
If in addition $d_1 =0 $, $d_3 \neq \pm 1 $  then acting with $e_- $ we can bring (\ref{cdneq0})  to the form (\ref{simpe3}) or to
\begin{equation}
X =e_3 .
\end{equation} 

Finally if $\alpha, b, c, d$ all vanish then we rescale to get
\begin{equation}
X =e_+ .
\end{equation} 

Thus, up to conjugation and rescaling,  the possible forms of $X$ are
\begin{align}
\label{1case} 
X_1  &= e_- + b  V_3 +  \gamma e_+ , \\
\label{2case} 
X_2 &= V_3 + c e_3  ,\quad  c \neq 0,\\
\label{3case} 
X_3 &= V_3 + d e_3^\ast , \quad d\neq 0,\\
\label{4case} 
X_4^\pm &= V_3 + c (e_3 \pm e_3^\ast ),\quad  c\neq 0,\\
\label{5case} 
X_5 &= V_3 + \gamma e_+ ,\\
\label{6case} 
X_6  &=e_3 ,\\
\label{7case} 
X_7 &= e_3^\ast ,\\
\label{8case} 
X_8^\pm  &= e_3 \pm e_3^\ast ,\\
\label{9case} 
X_9  &= e_3 + d_3 e_3^\ast + d_1 e_1^\ast , \quad d_1 \neq 0\\
\label{10case} 
X_{ 10 } &= e_+ ,
\end{align}
where any parameter can vanish unless otherwise specified.

Equations (\ref{xiepm}) -- (\ref{xivi}) give the KVF associated to  $X_i \in \mathfrak{g}  \rtimes \mathfrak{so }(3) $.  Using  (\ref{kvfprodddddd}) to compute their norms we find
\begin{equation*}
\begin{split} 
|\xi_{ X_1 }|^2 &=(1 + b^2 ) (x_1^2 + x_2 ^2)+ x_3^2  + 2 \gamma \geq 2 \gamma ,\\
|\xi_{ X_2}|^2 &= x_1^2 + x_2 ^2 + c^2 \cosh^2 x^->0 ,\\
|\xi_{ X_3}|^2 &= x_1^2 + x_2^2 + d^2 \sinh^2 x^- ,\\
|\xi_{ X_4^\pm }|^2 &= x_1^2 + x_2^2  + c^2 \mathrm{e}^{ \mp 2x^- }>0,\\
|\xi_{ X_5}|^2 &=  x_1^2 + x_2^2  ,\\
|\xi_{ X_6}|^2 &= \cosh^2 x^- >0,\\
|\xi_{ X_7}|^2 &= \sinh^2 x^- ,\\
|\xi_{ X_8^\pm }|^2 &=\mathrm{e}^{ \mp 2x^- }>0,\\
|\xi_{ X_9}|^2 &= (\cosh x^- - d_3 \sinh x^- )^2 + d_1^2  \sinh^2 x^- >0,\\
|\xi_{ X_{10}}|^2 &=0.
\end{split}
\end{equation*} 
The KVFs $X_3 $, $X_5 $, $X_7 $, $X_{ 10 }$ have zeros and need to be excluded, while in the case of $X_1 $ we need to impose $\gamma >0 $.
\end{proof} 
Making use of the representation  (\ref{rephe8}) it can be checked that the one-parameter groups associated to (\ref{1case1})--(\ref{9case9}) are all non-compact and thus have the topology of a line.

\subsection{Preserved SUSY}

\begin{proposition}\label{psusycw}
The KVFs of Proposition \ref{kvfcw} preserving a fraction $\nu>0$  of SUSY are given by Table \ref{frusycw}.
\begin{table}[htp]
\caption{Fraction $\nu>0$ of SUSY preserved by the KVFs of Proposition \ref{kvfcw}.}
\centering
\begin{tabular}{|c|c|c|}
\hline 
$X \in \mathfrak{g}  \rtimes \mathfrak{so }(3) $ & Condition & $\nu$  \\ \hline
$X_9 $ & - & $1/2 $\\ \hline
$ X_6 $ & - & $1/2 $\\\hline
$ X_8^\pm  $ & - & $1/2 $\\\hline
$X_1 $ & $4b^2 =9 \text{ or } 4b^2 =1$& $1/4 $\\\hline
\end{tabular}
\label{frusycw}
\end{table}
\end{proposition} 
\begin{proof}
Let us work at the point with coordinates, $x^i =x^\pm =0 $, where
\begin{equation}
\begin{split} 
\xi_{ e _\pm }&= \partial_\pm , \quad \mathrm{d} \xi ^\flat_{ e_ \pm } =0,\\
\xi_{ e_i } &= \partial_i , \quad \mathrm{d} \xi ^\flat_{ e_ i} =0, \\
\xi_{ e_i^\ast }&=0, \quad \mathrm{d} \xi ^\flat_{ e^\ast_ i} =2 \mathrm{d} x^i \wedge \mathrm{d} x^- ,\\
\xi_{V_i} &=0, \quad \mathrm{d} \xi_{V_i}^\flat = \epsilon_{ ijk }\mathrm{d} x_j \wedge \mathrm{d} x_k .
\end{split} 
\end{equation} 
Since $\varphi =\partial_+ $,  the $\beta$-term contribution to $L_\xi \epsilon $ is
\begin{equation}
\beta_\xi \epsilon = \frac{1}{4} \xi^i (\gamma_i \gamma_+ + 3 g_{ i + } )r \epsilon,
\end{equation} 
which is  non-zero only for $\xi_{e_i}$ and $\xi_-$,
\begin{equation}
\begin{split} 
\beta_{ \xi_{ e_i } } \epsilon &= \frac{1}{4} \gamma_i \gamma_+r \epsilon  ,\quad 
\beta_{ \xi_- } \epsilon = \frac{1}{4} (3 + \gamma_- \gamma_+ )r \epsilon .
\end{split} 
\end{equation} 
The other contribution comes from $\tfrac{1}{4} (\mathrm{d} \xi^\flat ) \cdot \epsilon  $, which gives
\begin{equation}
\begin{split} 
\frac{1}{4} \mathrm{d} \xi_{ e_i^\ast }^\flat \cdot \epsilon &
=\frac{1}{2} \gamma^i  \gamma^- \epsilon ,\\
\frac{1}{4} \mathrm{d}  \xi_{V_i} ^\flat \cdot \epsilon &
=\frac{\epsilon_{ i jk }}{4} \gamma^j  \gamma^k \epsilon .
\end{split} 
\end{equation} 
Putting all together
\begin{align}
L_{ \xi _{ e_+} } \epsilon &=0 ,\\
L_{ \xi _{ e_-} } \epsilon &= \frac{1}{4} (3 + \gamma_- \gamma_+ )r \epsilon ,\\ 
L_{ \xi _{ e_i} } \epsilon &=\frac{1}{4} \gamma_i \gamma_+ r \epsilon  ,\\
L_{ \xi _{ e^\ast _i} } \epsilon &
=\frac{1}{2} \gamma_i \gamma_+ \epsilon ,\\
L_{ \xi_{V_i} } \epsilon &=\frac{1}{4} \epsilon_{ ijk } \gamma_j \gamma_k \epsilon .
\end{align} 

For (\ref{1case1}) we have
\begin{equation}
\label{asdqweads} 
L_{ \xi_{X_1 } } \epsilon =0 \Leftrightarrow \left[ (3 + \gamma_-  \gamma_+  )r + 2b \gamma_1 \gamma_2 \right] \epsilon =0.
\end{equation} 
 Left multiplying by $\gamma_- $ gives the necessary condition
\begin{equation}
(3r + 2b \gamma_1 \gamma_2 )\gamma_- \epsilon =0.
\end{equation} 
Taking $ \epsilon_1, \epsilon_2   \in \operatorname{Ker} \gamma_- $  and substituting in (\ref{asdqweads}) gives
\begin{equation}
(r + 2b \gamma_1 \gamma_2 ) \epsilon =0
\end{equation} 
which has non-trivial solutions if and only if $b^2 =1/4 $, in which case we need 
\begin{equation}
\label{ccond1} 
\epsilon_1 \in \operatorname{Ker} \gamma_- \cap \operatorname{Ker} (i +2 b \gamma_1 \gamma_2 ),\quad 
\epsilon_2 \in \operatorname{Ker} \gamma_- \cap \operatorname{Ker} (-i + 2b \gamma_1 \gamma_2 ),
\end{equation} 
and $ \operatorname{Ker} \gamma _- \cap \operatorname{Ker} (\pm i  +  2b \gamma _1 \gamma _2 )$ is 1-dimensional, so $ \nu =\tfrac{1}{4}$. Taking instead $(3r + 2b \gamma_1 \gamma_2 ) \epsilon =0 $, which has non-trivial solutions if and only if $b^2 =9/4 $, and substituting in (\ref{asdqweads}) gives $  \gamma _- \gamma_+ \epsilon_i  =0 $ which has non-trivial solutions. So if $b^2 =9/4 $ we need  
\begin{equation}
\label{ccond2} 
 \epsilon_1  \in \operatorname{Ker} (3i +2 b \gamma_1 \gamma_2 ) \cap \operatorname{Ker} ( \gamma_- \gamma_+ ), \quad 
\epsilon_2 \in \operatorname{Ker} (-3i + 2b \gamma_1 \gamma_2 ) \cap \operatorname{Ker} ( \gamma_- \gamma_+ ),
 \end{equation}
and $\operatorname{Ker} (\pm 3i + 2b \gamma_1 \gamma_2 ) \cap \operatorname{Ker} ( \gamma_- \gamma_+ )$   is 1-dimensional, so again $\nu =\tfrac{1}{4}$.

A similar reasoning shows that (\ref{2case2}), (\ref{4case4})  preserve no SUSY and (\ref{6case6}), (\ref{8case8}) require $\epsilon_1 , \epsilon _2  \in \operatorname{Ker} \gamma_+   $ which is 2-dimensional, so $\nu =\tfrac{1}{2} $.
Finally (\ref{9case9}) gives
\begin{equation}
\label{hgtre43} 
(\gamma_3 r + 2d_3 \gamma_3 + 2d_1 \gamma_1) \gamma_+  \epsilon =0
\end{equation} 
which has 2-dimensional kernel, so $\nu =\tfrac{1}{2} $.
\end{proof}

\subsection{Geometry of the quotient}
All the KVFs of Proposition \ref{kvfcw} are spacelike with orbit homeomorphic to a line, so the quotient $M /\Gamma $ is in all cases a lorentzian 4-manifold with the topology of $\mathbb{R} ^4 $.

In order to determine the quotient metric we proceed as follows. Let $\xi$ be the KVF generating $\Gamma$, and pick a basis $ ( \chi_1 , \chi_2 , \chi_3 , \chi_4 )$ for $\xi^\perp $. Then for a KK geometry we have
\begin{equation}
\label{metrquotcomp} 
g = \frac{ \xi^\flat \otimes \xi^\flat }{g (\xi , \xi )} + \sum_{ i =1 }^4  C_{ ij }\chi ^\flat_i \odot \chi^\flat _j ,
\end{equation} 
where the coefficients $ C_{ ij } = C_{ ji } $ need to be determined, and the quotient metric $h$ is 
\begin{equation}
h =\sum_{ i =1 }^4  C_{ ij }\chi ^\flat_i \odot \chi^\flat _j 
\end{equation} 
provided that we re-express  $C_{ ij } $, $\chi_i^\flat $ in terms of coordinates $\tilde x_i $ well-defined on the quotient, i.e.~such that $\xi (\tilde x_i ) =0 $.
We carry out this procedure explicitly for the cases  $ X_1 $, $b =0 $, $X_6 $, $X_8^\pm, X_9  $ which have a larger isometry group, see Table \ref{geoquotcw} below.

If $b =0 $, $X_1 =  e_-  + \gamma e_+ $,
\begin{equation}
\xi_{ X_1 }= \partial_- + \gamma \partial_+ ,
\end{equation}
and
 $\xi^\perp = \operatorname{Span} (\partial_1 , \partial_2 , \partial_3 , \chi   )$, with
\begin{equation}
\chi =\partial_- - (\gamma + |x |^2 )\partial_+.
\end{equation} 
We find
\begin{equation}
\label{met1b0} 
h 
= \mathrm{d} x_1^2 + \mathrm{d} x_2^2 + \mathrm{d} x_3^2 - \frac{ \mathrm{d} u^2 }{2 \gamma + |x |^2 }  ,
\end{equation} 
with $\mathrm{d} u = \chi^\flat $, 
\begin{equation}
u =x^+ - \gamma  x^-.
\end{equation} 
Note that $ \xi (u) =\xi (x^1 )=\xi (x^2 )=\xi (x^3 )=0 $ so $(u , x^i )$ are well-defined coordinates on the quotient.
The isometry group of (\ref{met1b0}) is $O (3) \times \mathbb{R}  $.

For $X_6 =e_3 $,  
\begin{equation}
\xi_{ X_6 } = \cosh (x^- ) \partial_3  - x_3 \sinh (x^- ) \partial_+ ,
\end{equation}
 and
$ \xi^\perp = \operatorname{Span} (\partial_+ , \partial_1 , \partial_2 , \chi )$ with
\begin{equation} 
\chi =x_3 \sinh (x^- ) \partial_3 +  \cosh (x^- ) \partial_- .
\end{equation} 
We find
\begin{equation}
h =  \mathrm{d} x_1^2 + \mathrm{d} x_2^2 +  \mathrm{d} x^-  \left( C_1  \mathrm{d} x^-  + C_2  \chi^\flat \right) ,
\end{equation} 
\begin{equation}
C_1 =- |x |^2 - x_3^2  \tanh^2 (x^- ), \quad 
C_2 =\frac{2}{\cosh (x^- )}.
\end{equation} 
We can rewrite $h$ in terms of coordinates $(x_1 , x_2 ,  x^- , \tilde x^+ )$ well-defined on the quotient as
\begin{equation}\label{eq:nw-metric-1}
h = 2\mathrm{d} x^-  \mathrm{d}  \tilde x^+ + ( x_1^2 + x_2^2  ) (\mathrm{d} x^- )^2 +   \mathrm{d} x_1^2 + \mathrm{d} x_2^2,
\end{equation} 
where
\begin{equation}
\tilde x^+ =x^+ + \tfrac{x_3^2}{2} \tanh (x^- ).
\end{equation}

For $ X_8^\pm  = e_3\pm e_3^\ast $, 
\begin{equation} 
\xi_{ X_8^\pm}= \mathrm{e}^{\mp x^- }(\partial_3  \pm x_3 \partial_+ ) ,
\end{equation}
 and
$ \xi^\perp = \operatorname{Span} (\partial_+ , \partial_1 , \partial_2 , \chi )$ with
\begin{equation}
\chi = x_3 \partial_3\mp \partial_- .
\end{equation} 
We find
\begin{equation}
h = \mathrm{d} x_1^2 + \mathrm{d} x_2^2\mp 2  \mathrm{d} x^- \odot \chi^\flat  -  (x_1^2 + x_2^2 + 2 x_3^2 ) (\mathrm{d} x^- )^2 ,
\end{equation} 
which in terms of coordinates $(x^1 , x^2 , x^- , \tilde x^+ )$ well-defined on the quotient becomes
\begin{equation}\label{eq:nw-metric-2}
 h =2 \mathrm{d} \tilde  x^+\mathrm{d} x^-   +  (x_1^2 + x_2^2 ) (\mathrm{d} x^- )^2  + \mathrm{d} x_1^2 + \mathrm{d} x_2^2 ,
\end{equation} 
where
\begin{equation}
\tilde x^+ = x^+ \mp \frac{x_3^2 }{2} .
\end{equation} 
Note that quotienting along $X_8^+$ and along $X_8^-$ results in the same metric (\ref{eq:nw-metric-2}). This is  a 4-dimensional Cahen--Wallach space with quadratic form $ A =\eta^{ 0,2 }$. Note also that reduction along $X_6 $ and along $X_8^\pm$ results in the same quotient manifold, which makes intuitive sense since $X_6$ can be obtained from $X_8^\pm$ in the limit $x^- \rightarrow \pm\infty$ where $x^-$ is the parameter appearing in (\ref{conjacasy}).

For $X_9  $, taking for simplicity $d_3 =0 $, we have
\begin{equation}
\xi_{ X_9 } = \cosh x^-  \partial_3 - d_1 \sinh x^- \partial_1 + ( d_1 x_1  \cosh x^- - x_3 \sinh x^- ) \partial_+ ,
\end{equation} 
and $ \xi^\perp =\operatorname{Span} (\partial_+ , \partial_2 , \chi_1 , \chi_2 )$ with
\begin{equation}
\begin{split} 
\chi_1 &=d_1 \sinh x^- \partial_3 + \cosh  x^-\partial_1 ,\\
\chi_2 &=d_1 \sinh  x^- \partial_- + ( d_1 x_1  \cosh x^- - x_3 \sinh x^- )\partial_1 .
\end{split} 
\end{equation} 
In terms of coordinates $(\tilde x^+ , x^- , u, x^2 )$ well-defined on the quotient we have
\begin{equation}
h =2 \mathrm{d} \tilde x^+ \mathrm{d} x^- + x_2^2 ( \mathrm{d} x^- )^2 +  \frac{ (\mathrm{d} u -2u \coth (2 x^- ) \mathrm{d} x^- )^2 }{\cosh^2 x^-  + d_1^2 \sinh^2 x^- } + \mathrm{d} x_2^2   ,
\end{equation} 
where
\begin{equation}
\begin{split}
 u &=d_1 x^3  \sinh x^-  + x^1 \cosh x^- , \\
\tilde x^+ &=x^+  + \frac{x_1^2 }{2} \coth x^- + \frac{x_3^2 }{2} \tanh x^- .
\end{split}
\end{equation} 

The hereditary isometry algebra is
\begin{equation}  
\mathfrak{l }=\frac{ N_{X_\xi }(\mathfrak{g} \rtimes \mathfrak{so } (3)) }{ \mathbb{R}   X_\xi  },
\end{equation} 
where  $ N_{X_\xi }(\mathfrak{g}  \rtimes \mathfrak{so } (3)  )  $ is the normaliser of $X_\xi $ in $ \mathfrak{g}  \rtimes \mathfrak{so }(3) $.
For the KVFs of Proposition \ref{kvfcw} $ \mathfrak{l }$ is given in Table \ref{geoquotcw}. 
\begin{table}[htp]
\caption{Hereditary isometry algebra $\mathfrak{l }$ of the KK reductions of $\mathrm{CW}_5$  by the KVFs of Proposition \ref{kvfcw}.}
\centering
\begin{tabular}{|c|c|c|}
\hline 
$X \in \mathfrak{g}  \rtimes \mathfrak{so }(3) $ & $\mathfrak{l } $ & generators \\ \hline
$X_1 , \, b =0  $& $\mathbb{R}  \oplus \mathfrak{so }(3) $& $e_+ ; V_1 , V_2 , V_3$ \\ \hline
$X_1, \,  b \neq 0 $&$\mathbb{R}  \oplus \mathfrak{so }(2) $ & $e_+ ; V_3$ \\ \hline
$X_2  $&$\mathbb{R}  \oplus \mathfrak{so}(2) $ & $e_+ ; V_3$ \\ \hline
$X_4^\pm  $&$\mathbb{R}  \oplus \mathfrak{so}(2) $ & $e_+ ; V_3 $ \\ \hline
$X_6 $ & $\mathfrak{h}  \rtimes \mathfrak{so}(2)  $ & $e_+ , e_2 ,  e_2^\ast , e_1 , e_1^\ast ; V_3 $\\ \hline
$X_8^\pm  $ & $\mathfrak{h}  \rtimes (\mathbb{R}  \oplus \mathfrak{so}(2) )$ & $e_+ ,  e_2 ,  e_2^\ast , e_1 , e_1^\ast ; e_- ,V_3 $\\ \hline
$ X_9 $& $\mathfrak{h}  $ & $e_+ , e_2 , e_2^\ast, e_1  + d_1 e_3^\ast , e_1^\ast $ \\\hline
\end{tabular}
\label{geoquotcw}
\end{table}
The  Lie algebra $\mathfrak{h}   $ is the 5-dimensional Heisenberg algebra generated by $\{e_1 , e_1^\ast , e_2 , e_2^\ast , e_+ \} $ in the case of $X_6 $, $X_8^\pm$, and by $\{e_1 + d_1e_3^\ast , e_1^\ast , e_2 , e_2^\ast , e_+ \} $ in the case of $X_9$. Note that  while the quotient by $X_6$ is a  CW space, cfr.~equation (\ref{eq:nw-metric-1}), the hereditary isometry algebra associated to $X_6$ is only a proper subalgebra of the CW isometry algebra. Therefore this is an example where the quotient has additional accidental symmetry.

\section{Conclusions and Summary}
\label{sec:conc-summ}

In this paper we have classified four-dimensional Kaluza--Klein
reductions of certain supersymmetric five-dimensional lorentzian
geometries found in \cite{Beckett:2021cwx}; namely,
\begin{enumerate}
\item $-\RR \times S^4$;
\item $\RR \times \AdS_4$; and
\item a conformally flat Cahen--Wallach symmetric space.
\end{enumerate}
We have concentrated on reductions leading to four-dimensional
lorentzian or riemannian manifolds.  Although we do not consider them
in this paper, the question of null reductions is interesting in the
context of non-relativistic supersymmetry and studying the null
reductions of the above backgrounds might give four-dimensional
supersymmetric Newton--Cartan geometries different from those in
\cite{Figueroa-OFarrill:2019ucc} in a similar way to how
three-dimensional supersymmetric Newton--Cartan geometries can be
obtained via null reduction of four-dimensional supersymmetric
geometries \cite{Bergshoeff:2020baa}.

For the three backgrounds listed above, we list the possible
one-parameter subgroups of isometries resulting in a lorentzian or
riemannian quotient, identify the hereditary isometries of the four-dimensional
quotient, the fraction of the supersymmetry which is preserved and, in
most cases, the form of the metric in the quotient. The possible
generators of one-parameter subgroups are given in
Proposition~\ref{kvfsn} for $-\RR \times S^4$,
Proposition~\ref{kvfads} for $\RR \times \AdS_4$ and
Proposition~\ref{kvfcw} for the Cahen--Wallach spacetime. The
hereditary isometries of the quotients are listed in Table~\ref{s4quotiso} for
$-\RR \times S^4$, Table~\ref{adsgeneratorslie} for
$\RR \times \AdS_4$ and Table~\ref{geoquotcw} for the Cahen--Wallach
space. The conditions for preservation of supersymmetry and the
fraction of supersymmetry which is preserved upon reduction are
described in Proposition~\ref{psusys4} for $-\RR \times S^4$ and
listed in Table~\ref{presusyads} for $\RR \times \AdS_4$ and
Table~\ref{frusycw} for the Cahen--Wallach space.

It is worth highlighting some of the half-BPS reductions; that is,
those which preserve half of the supersymmetry.  Firstly, there are
four half-BPS lorentzian reductions: one (labelled $\lambda_4$ with
parameter $\beta = 2\|\varphi\|$) of $\RR \times \AdS_4$ and three
(labelled $X_6, X_8^\pm, X_9$) of the Cahen--Wallach space.  These
give four-dimensional lorentzian geometries admitting an $N=1$
supersymmetry algebra.  It is then a natural question to ask whether
they are contained in the classification of \cite{deMedeiros:2016srz}.
The geometries in that paper have supersymmetry algebras which are
filtered deformations of graded maximally supersymmetric subalgebras
of the $N=1$ Poincar\'e superalgebra and consist of Minkowski spacetime,
$\AdS_4$, $-\RR \times S^3$, $\RR \times \AdS_3$ and the Nappi--Witten
group $\NW_4$.  They all share the property that the metric is
conformally flat.  The reductions labelled $X_6$ and $X_8^\pm$ of the
Cahen--Wallach space are isometric to the Nappi--Witten group, as
shown by the explicit form of the quotient metric in
equations~\eqref{eq:nw-metric-1} and \eqref{eq:nw-metric-2}.  The
other reduction $(X_9)$ of the Cahen--Wallach space depends on two
parameters $d_1 \neq 0$ and $d_3$.  We have calculated the Weyl
curvature tensor in the case   $d_3 = 0$ and found it to be non-vanishing, so that this reduction is not conformally flat, in
contrast to the backgrounds in \cite[Thm.~14]{deMedeiros:2016srz}.
We have not calculated the curvature tensor for any $d_3 \neq 0$, but
it seems likely that they are not conformally flat reductions either.  Finally,
the half-BPS lorentzian reduction of $\RR \times \AdS_4$
is novel to the best of our knowledge.  Indeed, calculating the
Weyl curvature tensor of the metric in
equation~\eqref{eq:new-lor-4d-metric-ads} one sees that it too is non-vanishing.

\section*{Acknowledgements}
GF would like to thank the Simons Foundation for its support under the Simons Collaboration on Special Holonomy in Geometry, Analysis and Physics (grant 488631).

\appendix
\section{Cahen--Wallach spaces}
\label{cwgener} 

Cahen--Wallach spaces are  locally symmetric lorentzian manifolds which exist for any dimension $n  \geq 3 $. 
We first construct them as symmetric spaces, following \cite{Figueroa-OFarrill:2001hal} for the most part,  and then derive a coordinate expression for their metric.

 Let $V$ be an $(n-2)$-dimensional vector space with basis $ (e_i )$, $V^\ast $ its dual with basis $(e_i^\ast )$, $Z=\operatorname{Span} (e_+ )$, $Z^\ast = \operatorname{Span} (e_- ) $. Let $A $ be a symmetric bilinear form on $V$. Then the Lie algebra
 \begin{equation}
\label{jksdhkjds} 
\mathfrak{g}   =V \oplus V^\ast \oplus Z \oplus Z^\ast 
 \end{equation}
  with non-zero  Lie brackets
\begin{equation}
[ e_- , e_i ] = e_i^\ast , \quad [ e_- , e_i^\ast ]= A_{ ij }e_j , \quad [ e_i^\ast , e_j ] = A_{ ij }e_+ ,
\end{equation} 
is solvable since $\mathfrak{g} ^\prime = \operatorname{Span}  \left( e_i^\ast , e_i , e_+ \right) $, $\mathfrak{g} ^{ \prime \prime } =\operatorname{Span} (e_+ ) $, $\mathfrak{g} ^{ \prime \prime \prime } =0 $. Set
\begin{equation}
\mathfrak{h}  = \operatorname{Span} (e_i^\ast ), \quad \mathfrak{k}  = \operatorname{Span} (e_i , e_+ , e_- ).
\end{equation} 
Then $\mathfrak{h}  $ is an Abelian subalgebra of $\mathfrak{g}$, $[ \mathfrak{h}, \mathfrak{k}  ] \subset \mathfrak{k}  $, $[ \mathfrak{k}  , \mathfrak{k}  ]= \mathfrak{h}  $. Thus $\mathfrak{g}  =  \mathfrak{h}  \oplus \mathfrak{k}   $ is a symmetric splitting of $\mathfrak{g}$. Note that $  \{ e_i , e_i^\ast , e_+ \} $ generates a $(2n + 1 )$-dimensional Heisenberg algebra. Let $G$ (respectively $H$) be the unique simply connected Lie group with Lie algebra $\mathfrak{g}$  ($ \mathfrak{h}  $). Then $M = G / H $ is a locally symmetric space. Denote by $o$ the identity coset, $o =e H $.

Let $B$ be the symmetric bilinear form on $V \oplus Z \oplus Z^\ast \simeq  \mathfrak{g}  / \mathfrak{h}  $    with non-zero components
\begin{equation}
\label{Bform}
B (e_i , e_j ) = \delta_{ ij }, \quad 
B (e_+, e_- )=1.
\end{equation} 
Note that $B$ is invariant under the action of $H$ by conjugation. Since $H$ is Abelian we can simply check that $h =\exp (c e_i^\ast ) $ acts by isometries. Using (\ref{eistej}), (\ref{eisteminus}) we have e.g.
\begin{equation}
\begin{split} 
B (h e_- h^{-1} , h e_- h^{-1} ) &
= B \left( e_- - c A_{ ij }e_j -\frac{  c^2}{2}  A^2_{ ii } e_+, e_- - c A_{ ik }e_k - \frac{ c^2}{2}A^2_{ ii }e_+ \right) \\ &
=-c^2  A^2_{ ii } +c^2   A_{ ij }A_{ ik } \delta_{ jk } =0 =B (e_- , e_- ).
\end{split} 
\end{equation} 
The other components can be checked similarly. 

Define a lorentzian $G$-invariant metric $\beta$ on $G $ by setting, for  any $g \in G  $,
\begin{equation}
\beta_g (U , V ) = B ( g^{-1} U , g^{-1}V ).
\end{equation} 
The left action of $H$ on $T M $ corresponds to the action of $ \operatorname{Ad}_h $ on $\mathfrak{g}  / \mathfrak{h}  $. Since $B$ is $H $-invariant, $\beta$ descends to a  well defined  metric on $M$, which we still denote by $\beta$,
\begin{equation}
\label{metrickhjsdaoiuqew} 
\beta_{g \cdot o } (U , V ) = B ( g^{-1} U , g^{-1}V ),
\end{equation} 
where we have identified $ T_o M $ with $\mathfrak{g}  $, and more generally we identify $ T_{ g \cdot o }M $ with $ T_g G $. 

We now derive a coordinate expression for $\beta$. On $ \mathbb{R} ^n $ take coordinates $x^1  , \ldots, x^{ n-2}$, $x^+$, $x^- $ and define
\begin{equation}
\label{modexpcoords} 
\begin{split} 
\sigma &: \mathbb{R} ^n \rightarrow G , \quad  ( x^i , x^+ , x^- ) \mapsto \exp (  x^+ e_+ + x^- e_- ) \exp ( x^i e_i ).
\end{split} 
\end{equation} 
The map $\sigma$ provides \emph{modified exponential coordinates} on $G$ and, acting on $o$, on $M$.
The curve $ ( x^1, \ldots , x^{ n-1 }, x^+ , x^- + t )$ on $\mathbb{R} ^n $  has tangent $ \partial / \partial x^- $, thus
\begin{equation}
\label{sigmaminus} 
\sigma_\ast (\partial_- ) = \ddt \sigma (x^i , x^+ , x^- + t ) = e_- \sigma (x).
\end{equation} 
Note that
\begin{equation}
\begin{split} 
\sigma (x)^{-1} e_- \sigma (x) &
=\exp (- x^j e_j )e_- \exp ( x^k e_k )
=\exp (-x^k  \operatorname{ad}_{ e_k } )(e_- )\\ &
=e_-  + x^i  e_ i^\ast  + \frac{1}{2} x^i x^j A_{ ij }e_+ .
\end{split} 
\end{equation} 
Similarly 
\begin{align}
\label{sigmaplus} 
\sigma_\ast (\partial_+ ) &=e_+ \sigma (x) =\sigma (x) e_+,\\
\label{sigmai} 
\sigma_\ast (\partial_i ) &= \sigma (x) e_i .
\end{align} 
Pulling back (\ref{metrickhjsdaoiuqew}) by $\sigma $ we have
\begin{equation}
(\sigma^\ast \beta)_x (U ,V ) 
= \beta_{ \sigma  (x) } (\sigma _\ast U , \sigma  _\ast V )
= B ( \sigma (x)^{-1} ( \sigma_\ast U ), \sigma (x)^{-1}  (\sigma _\ast V) ).
\end{equation} 
Writing $g =\sigma^\ast \beta $ and taking into account that terms in $e_i^\ast $ are zero on the quotient space we  find
\begin{equation}
\begin{split} 
g_x ( \partial_i , \partial_j ) &= \delta_{ ij },\\
g_x (\partial_+ , \partial_- ) &= B ( e_+ , \sigma (x)^{-1} e_- \sigma (x)  ) = B (e_+ , e_- ) =1,\\
g_x (\partial_- , \partial_- ) &= B (\sigma (x)^{-1} e_- \sigma (x) , \sigma (x)^{-1} e_- \sigma (x)  )
=x^i x^j A_{ ij },\\
g_x (\partial_i , \partial_{ \pm } )&= g_x (\partial_+ , \partial_+ ) =0.
\end{split} 
\end{equation} 
Therefore, with respect to the modified exponential coordinates (\ref{modexpcoords}),  the $G$-invariant metric (\ref{metrickhjsdaoiuqew}) on $M$ is
\begin{equation} 
\label{cwametric} 
g = 2 \mathrm{d} x^+ \mathrm{d} x^- + A_{ ij } x^i x^j \, (\mathrm{d} x^- )^2  +  \delta_{ ij }\mathrm{d} x^i \mathrm{d} x^j .
\end{equation} 

 The vector field $\varphi =\partial_+ $ is null and parallel.
All the coordinates $ x^i , x^\pm $ range in $\mathbb{R}  $ and the resulting space is complete.  
It is known that a CW space is indecomposable if and only if $A$ is non-degenerate.
The only non-zero components of Riemann and Ricci tensors Ricci tensor in the coordinates (\ref{cwametric}) are
\begin{equation}
R_{ -i-j }=-  A_{ ij }, \quad 
 \operatorname{Ric} _{ -- } = - \operatorname{Tr} A,
\end{equation} 
hence CW spaces are scalar-flat. The non-vanishing components of the Weyl tensor are
\begin{equation} 
\label{weylcw} 
W_{-i-j} = - A_{ ij } +  \frac{1}{n-2}  \operatorname{Tr} A \,   \eta^{ 0,n-2}_{ ij },
\end{equation} 
so a CW space is conformally flat if and only if  $A = a \eta^{ 0, n-2 } $ for some constant $a$.

\subsection{Killing vector fields}
\label{cwkvfasdlk} 
It is clear from the symmetric space description that a CW space has isometry algebra 
\begin{equation} 
\mathfrak{g}  \rtimes \mathfrak{s},
\end{equation} 
where
\begin{equation}
\mathfrak{g}  =V \oplus V^\ast \oplus Z \oplus Z^\ast 
\end{equation} 
as in (\ref{jksdhkjds})  and  
\begin{equation}
\label{defssss} 
 \mathfrak{s } =\{ s \in \mathfrak{ so }(V) : s^T A + A s=0 \} 
 \end{equation}
is  the subalgebra of $ \mathfrak{so }(V) $ leaving $A$ invariant. 
The action of $\mathfrak{s} $ on $\mathfrak{g}  $ is the natural action of $\mathfrak{so }(V) $ on $V$,  the adjoint action on $ V^\ast $, and the trivial action on $Z$, $Z^\ast $. Such an action preserves brackets. In fact for $s \in \mathfrak{s} $ we have e.g. 
\begin{equation}
\begin{split} 
s \cdot [ e_- , e_i ] &= [ s \cdot e_- , e_i ] + [ e_- , s \cdot e_i ] 
=0 +  [ e_- ,  s^j_{ \ i } e_j]
= s^j_{ \ i } e_j^\ast 
=-  s ^i_{ \ j } e_j^\ast 
=s \cdot e_i^\ast ,\\
s \cdot [ e_- , e_i^\ast  ] &
=  [ e_- ,  -s^i_{ \ j } e_j^\ast ]
=-  s ^i_{ \ j } A_{ jk }e_k =- A_{ ik }s^k_{ \ j }e_k = s \cdot A_{ik}e_k ,
\end{split} 
\end{equation}  
having used  $s^T =- s $, $s ^T A =- A s $.
We can therefore form the semidirect product $ \mathfrak{g}  \rtimes \mathfrak{s} $ with, for $s_i \in \mathfrak{s} $, $g_i \in \mathfrak{g}  $, bracket
\begin{equation}
\label{sgbra} 
[ ( s_1 , g_1 ), (s_2 , g_2 )]
=([ s_1 , s_2 ], [ g_1 , g_2 ] + s_1 \cdot g_2 - s_2 \cdot g_1 ).
\end{equation} 

We now want to write down the Killing vector fields corresponding to the generators $ (e_i , e^\ast_i , e_+ , e_- )$ of $\mathfrak{g}$   in terms of the modified exponential coordinates $(x^i , x^+ , x^- ) $.  We define an action $x \mapsto g \cdot x $ of $G$ on the exponential coordinates so that (\ref{modexpcoords}) is $G$-equivariant,
\begin{equation}
g \sigma (x) \cdot o = \sigma (g \cdot x ) \cdot o.
\end{equation} 
It follows $ \sigma (g \cdot x )^{-1} g \sigma (x) \in H $ as it fixes $o$. Thus, taking  $ X \in \mathfrak{g}  $, $g =\exp (t X )$ to be the corresponding one-parameter subgroup of $G$, we can write
\begin{equation}
\exp (t X ) \sigma (x) = \sigma (\exp (t X ) \cdot x )  h (t, X )
\end{equation} 
for some $h \in H $ depending on $t, X $. Clearly  $ h (0,X )= e $ for any $X$. Differentiating and evaluating at $t =0 $ we get
\begin{equation}
\label{mainkvfeq} 
 X \sigma (x)  = \sigma_\ast |_x (\xi_X ) + \sigma (x)  Y 
\end{equation} 
where $\xi_X $ is the KVF corresponding to $X \in \mathfrak{g}  $ with respect to modified exponential coordinates $x$ and $Y \in \mathfrak{h}  $.

As calculated before, see (\ref{sigmaplus}), (\ref{sigmaminus}), (\ref{sigmai}),
\begin{align} 
\label{kvfpm} 
\sigma_\ast ( \partial_\pm |_x ) &= e_\pm \sigma (x) ,\\
\sigma_\ast (\partial_i |_x )&= \sigma (x) e_i .
\end{align} 
Comparing (\ref{kvfpm})  with (\ref{mainkvfeq}) we immediately see that
\begin{equation}
\xi_{ e_\pm } |_x 
= \partial_{ \pm } |_x .
\end{equation} 

For $\xi_{ e_i }|_x $ we calculate, writing $\sigma $ for $\sigma (x) $,
\begin{equation}
\begin{split}
e_i \sigma &
= \sigma  \sigma^{-1} e_i \sigma 
= \sigma (\exp (- x^k e_k ) \exp (- x^- e_- )e_i \exp ( x^- e_- )\exp (x^k e_k )  )\\  &
=\sigma (\exp (- x^k \operatorname{ad}_{ e_k } ) )(\exp (- x^- \operatorname{ad}_{ e_- } )(e_i ))
=\sigma (\exp (- x^k \operatorname{ad}_{ e_k } ) )\left( C e_i - \frac{S }{\sqrt{ |\lambda_i |}  } e_i^\ast\right) \\ &
=\sigma \left( C e_i - \frac{S }{\sqrt{ |\lambda_i |}} (e^\ast_i + x^i \lambda_i e_+  )\right) 
=\sigma_\ast \left( C \partial _i  - \mathrm{sign} (\lambda_i ) \sqrt{ | \lambda_i |}S x^i \partial_+ \right)- \frac{S }{\sqrt{| \lambda_i |}} e_i^\ast ,
\end{split}
\end{equation} 
having used the equations given in Section \ref{conjac} and with $C$, $S$  given by (\ref{cands}).
It follows
\begin{equation}
\xi_{ e_i }|_x = C \partial_i |_x - \operatorname{sign} (\lambda_i )\sqrt{ |\lambda_i |} S x^i  \partial_+ |_x \quad \text{(no sum over $i$)}  .
\end{equation} 

Finally for $\xi_{ e_i^\ast }|_x  $ we have, again writing $\sigma$ for $\sigma (x) $,
\begin{equation}
\begin{split}
e_i^\ast \sigma  &
=\sigma \sigma^{-1} e_i^\ast \sigma 
=\sigma \left( \exp (- x^k \operatorname{ad}_{ e_k } ) (\exp (- x^- \operatorname{ad}_{ e_- } )(e_i^\ast ) \right) \\ &
=\sigma \left( \exp (- x^k \operatorname{ad}_{ e_k } ) (C e_i^\ast - \operatorname{sign} (\lambda_i ) \sqrt{ | \lambda _i |} S e_i  ) \right) \\ &
=\sigma \left( C e_i^\ast + C x^i \lambda_i  e_+ -\operatorname{sign} (\lambda_i ) \sqrt{ |\lambda_i |} S e_i  \right) \\ &
=\sigma ( C e_i^\ast ) + \sigma_\ast  \left(  C x^i \lambda_i  \partial _+ -\operatorname{sign} (\lambda_i ) \sqrt{ |\lambda_i |} S \partial _i  \right) ,
\end{split}
\end{equation} 
thus
\begin{equation}
\xi_{ e_i^\ast } |_x 
= C x^i \lambda_i \partial_+ |_x - \operatorname{sign} (\lambda_i )\sqrt{ |\lambda_i |} S \partial_i |_x  \quad \text{(no sum over $i$)}.
\end{equation} 

In summary
\begin{align} 
\xi_{ e_ \pm } &= \partial_{ \pm },\\
\xi_{ e_i }&= C \partial_i - \operatorname{sign} (\lambda_i )\sqrt{ |\lambda_i |} S x^i  \partial_+  \quad \text{(no sum over $i$)},\\
\xi_{ e_i^\ast } &
= C x^i \lambda_i \partial_+ - \operatorname{sign} (\lambda_i )\sqrt{ |\lambda_i |} S \partial_i   \quad \text{(no sum over $i$)}.
\end{align} 
Note that $\xi_{ e_i^\ast }$ vanishes for $x =0 $, as it must.

Let us consider in more detail the case relevant for the paper,  $n =5 $, $ A = a \eta^{0,3 }$, $a>0 $. In this case $ \mathfrak{s } =\mathfrak{so } (3) $ and we choose generators $(V_1 , V_2 , V_3 )$ satisfying
\begin{equation}
[ V_i, V_j ] = -\epsilon_{ ijk }V_k 
\end{equation} 
so that the associated KVFs $ \xi_{ V_i }$ are the usual  generators of rotations in $\mathbb{R} ^{ 0,3 }$,
\begin{equation}
\xi_{ V_i } = \epsilon_{ ijk } x_j \partial_k .
\end{equation}
By (\ref{sgbra}) the non-trivial brackets of $\mathfrak{so } (3) $ with $\mathfrak{g}$   are
\begin{equation}
\label{hhvvvvi} 
[ V_i , e_j ] =  V_i e_j = -\epsilon_{ ijk }e_k, \quad [ V_i , e^\ast_j ] =  V_i e^\ast _j =- \epsilon_{ ijk }e^\ast_k.
\end{equation}

For use in the main text, we give here the inner products with respect to the metric (\ref{cwametric}) of the KVFs $(\xi_{ e_i }, \xi_{ e_i^\ast }, \xi_{ e^\pm }, \xi_{ V_i })$ for $n =5 $, $ A =a \eta^{ 0,3 }$, $a >0 $.
 \begin{equation}
 \label{kvfprodddddd} 
 \begin{split}
 |\xi_{ e_+ }|^2 &= \langle \xi_{ e_i }, \xi_{ e_+ } \rangle  =\langle \xi_{ e_i^\ast }, \xi_{ e_+ } \rangle = \langle\xi_{ V_i} , \xi_{ e_\pm } \rangle = 0,\\
 \langle \xi_{ e_+ }, \xi_{ e_- }\rangle &=1,\\
 |\xi_{ e_- }|^2 &= a |x |^2 ,\\
 \langle \xi_{ e_i }, \xi_{ e_- }\rangle &= - \sqrt{ a }x_i \sinh  (\sqrt{a} x^-  ),\\
  \langle \xi_{ e^\ast _i }, \xi_{ e_- }\rangle &=  a x_i \cosh  (\sqrt{a} x^-  ),\\
 \langle \xi_{ e_i }, \xi_{ e_j  }\rangle &=\cosh^2 (\sqrt{a} x^-  )\delta_{ ij },\\
 \langle \xi_{ e_i }, \xi_{ e^\ast_j }\rangle &=- \sqrt{ a } \sinh(\sqrt{a} x^-  ) \cosh (\sqrt{a} x^-  ) \delta_{ ij },\\ 
  \langle \xi_{ e^\ast _i }, \xi_{ e_j^\ast }\rangle &=a \sinh^2 (\sqrt{a} x^-  )\delta_{ ij },\\
  \langle \xi_{V_i} ,\xi_{  V_j} \rangle &= |x |^2 \delta_{ ij } - x_i x_j ,\\
\langle \xi_{V_i} , \xi_{ e_j }\rangle &= - \cosh(\sqrt{a} x^-  ) x_k \epsilon_{ kij },\\
\langle \xi_{V_i} , \xi_{ e^\ast_j }\rangle &=  \sqrt{ a }\sinh(\sqrt{a} x^-  ) x_k \epsilon_{ kij }.
 \end{split}
 \end{equation} 
In particular $ \xi_{ e_i }$ is spacelike, $ \xi_{ e_+ }$ is null, and  $\xi_{ e_i^\ast }$, $\xi_{V_i} $, $\xi_{ e_- }$ have non-negative norm.

\subsection{Inner automorphisms of $\mathfrak{g }\rtimes \mathfrak{s} $}
\label{conjac} 
For use in the main text and general reference we record the following expressions,
\begin{align}
\operatorname{Ad}_{   \exp \left(   c_k e_k \right)  } (e_i^\ast )
&= e_i^\ast  -  A_{ ik } c_k  e_ +  ,\\
\operatorname{Ad}_{   \exp \left(  c_k e_k \right)  } (e_- )
&= e_-  -   c_k  e_ k^\ast  + \frac{1}{2} c_i c_j A_{ ij }e_+   ,\\
\operatorname{Ad}_{   \exp \left(   c_k e_k \right)  } (e_i   )& =e_i ,\\
\operatorname{Ad}_{   \exp \left(   c_k e_k \right)  } (e_ +    ) &=e_ +  ,\\
\label{eistej} 
\operatorname{Ad}_{   \exp \left(   c_k e_k^\ast  \right)  } (e_i  )
&=e_i + c_k  A_{  ik } e_+ ,\\
\label{eisteminus} 
\operatorname{Ad}_{   \exp \left(   c_k e^\ast _k \right)  } (e_- )
&=e_- -  c_kA_{ kj }  e_j - \frac{1}{2}  c_i c_k A_{ ij } A_{jk } e_+ ,\\
\operatorname{Ad}_{   \exp \left(   c_k e^\ast _k \right)  } (e^\ast_i   )& =e^\ast_i ,\\
\operatorname{Ad}_{   \exp \left(   c_k e_k^\ast  \right)  } (e_ +    ) &=e_ + .
\end{align} 

Furthermore, assuming that $A $ has been diagonalised, 
\begin{equation}
A_{ ij } =\lambda_i \delta_{ ij },
\end{equation} 
for $\lambda_i \neq 0 $ we have
\begin{align}
\label{aaaa1} 
\exp ( x^-  \operatorname{ad}_{ e_- } )(e_i )  &= C e_i  + \frac{ S}{\sqrt{| \lambda_i |} }e_i^\ast ,\\
\label{aaaa2} 
\exp ( x^-  \operatorname{ad}_{ e_- } )(e_i^\ast )  &= C e_i ^\ast + \operatorname{sign} ( \lambda_i )\sqrt{| \lambda_i |} S e_i  ,
\end{align} 
where 
\begin{equation}
\label{cands} 
C =\begin{cases}
\cos ( \sqrt{ -\lambda_i  } x^- )&\text{if $\lambda_i <0 $},\\
\cosh ( \sqrt{ \lambda_i  } x^- )&\text{if $\lambda_i >0 $},
\end{cases} \qquad 
S =\begin{cases}
\sin  ( \sqrt{ -\lambda_i  } x^- )&\text{if $\lambda_i <0 $},\\
\sinh ( \sqrt{ \lambda_i  } x^- )&\text{if $\lambda_i >0 $}.
\end{cases} 
\end{equation} 
For $\lambda_i =0 $ instead 
\begin{equation}
\exp ( x^{-}  \operatorname{ad}_{ e_- } ) (e_i ) =e_i +  x^{-}  e_i^\ast , \quad 
\exp ( x^{-}  \operatorname{ad}_{ e_- } ) (e_i^\ast  ) =e_i^\ast ,
\end{equation} 
so (\ref{aaaa1}), (\ref{aaaa2}) still hold provided that we extend the definition of $C$ and $S$ by setting
\begin{equation}
C =\begin{cases}
1 &\text{if $\lambda_i =0 $} ,\\
\cos ( \sqrt{ -\lambda_i  } x^- )&\text{if $\lambda_i <0 $},\\
\cosh ( \sqrt{ \lambda_i  } x^- )&\text{if $\lambda_i >0 $},
\end{cases} \qquad 
S =\begin{cases}
0 \ \text{and } S / \sqrt{ |\lambda_i |} \rightarrow x^-  &\text{if $\lambda_i =0 $} ,\\
\sin  ( \sqrt{ -\lambda_i  } x^- )&\text{if $\lambda_i <0 $},\\
\sinh ( \sqrt{ \lambda_i  } x^- )&\text{if $\lambda_i >0 $}.
\end{cases} 
\end{equation} 

Using (\ref{hhvvvvi}) we also have
\begin{align} 
\operatorname{Ad}_{   \exp \left(  \sum_k c_k e_k \right)  } (V_i   )&
= V_i  -  c_k  \epsilon_{ kij } e_j  ,\\
\operatorname{Ad}_{   \exp \left(  \sum_k c_k e_k^\ast \right)  } (V_i   )&
= V_i  -  c_k  \epsilon_{ kij } e_j^\ast ,
\end{align} 
and, for $(i,j,k )$ a permutation of $(1,2,3 )$,
\begin{align} 
\operatorname{Ad}_{   \exp \left(  t V_i  \right)  } (V_j    )& =\cos t V_j  \mp   \sin  t V_k , \\
\operatorname{Ad}_{   \exp \left(  t V_i  \right)  } (e_j    )& =\cos t e_j  \mp   \sin  t e_k , \\
\operatorname{Ad}_{   \exp \left(  t V_i  \right)  } (e_j^\ast    )& =\cos t e_j^\ast  \mp   \sin  t e_k^\ast ,
\end{align} 
with the upper (lower) sign for an even (odd) permutation.

\section{Gamma matrices}
\label{gmatr} 

It is useful to have an
explicit representation of $\Cl(1,4) $ to compute the fraction of preserved SUSY. Let
$(\sigma_1 , \sigma_2 , \sigma_3 )$ be the Pauli matrices, $I_2 $ the
$2$-dimensional identity matrix.

In Section \ref{psusys4} we used
\begin{equation}
\label{mainclrep} 
\begin{split}
\gamma_0 &=\sigma_3 \otimes I_2   
=\left(\begin{array}{c|c} I_2  & 0  \\\hline 0 & -I_2 \end{array}\right),\\
\gamma_1 &= (i \sigma_2 )\otimes \sigma_1 
=\left(\begin{array}{c|c}0 & \sigma_1  \\\hline - \sigma_1 & 0\end{array}\right),\\
\gamma_2 &= (i \sigma_2 )\otimes \sigma_2 
=\left(\begin{array}{c|c}0 & \sigma_2  \\\hline - \sigma_2 & 0\end{array}\right),\\
\gamma_3 &= (i \sigma_2 )\otimes \sigma_3 
=\left(\begin{array}{c|c}0 & \sigma_3  \\\hline - \sigma_3 & 0\end{array}\right),\\
\gamma_4 &= (i \sigma_1 )\otimes I_2  
=\left(\begin{array}{c|c}0 & i I_2   \\\hline i I_2  & 0\end{array}\right),
\end{split}
\end{equation}
with $\gamma_0 =\gamma_1 \gamma_2 \gamma_3 \gamma_4  $ corresponding to the timelike direction and $(\gamma_1 , \gamma_2 , \gamma_3 , \gamma_4 )$ to the spacelike ones.

In Section \ref{psusyads4} for
 \begin{equation}
 \lambda_5 =R_{ 13 }- R_{ 34 }
 \end{equation}
we chose generators $(\gamma_1, \gamma_y , \gamma_3 , \gamma_4 , \gamma_5 )$, with $\gamma_1 $ corresponding to the timelike direction, and used
\begin{equation}
\label{rr1} 
\begin{split}
\gamma_1 &=( \sigma_3 )\otimes  I_2 
,\\
\gamma_y &=(i \sigma_2  )\otimes  \sigma_1 
,\\
\gamma_3 &=(i \sigma_2 )\otimes  \sigma_2 
,\\
\gamma_4 &=(i \sigma_2 )\otimes  \sigma_3 
,\\
\gamma_5 &=(i \sigma_1 )\otimes  I _2 .
\end{split}
\end{equation} 
For all the other KVFs we took generators  $(\gamma_2, \gamma_y , \gamma_3 , \gamma_4 , \gamma_5 )$, with $\gamma_2 $ corresponding to the timelike direction, and used
\begin{equation}
\label{rr2} 
\begin{split}
\gamma_2 &=( \sigma_3 )\otimes  I_2 
,\\
\gamma_y &=(i \sigma_2  )\otimes  \sigma_1 
,\\
\gamma_3 &=(i \sigma_2 )\otimes  \sigma_2 
,\\
\gamma_4 &=(i \sigma_2 )\otimes  \sigma_3 
,\\
\gamma_5 &=(i \sigma_1 )\otimes  I _2 .
\end{split}
\end{equation} 

Note that (\ref{mainclrep}), (\ref{rr1}), (\ref{rr2}) only differ by a relabelling of the gamma matrices.

In Section \ref{psusycw} we used the  representation  (\ref{mainclrep}). Additionally we defined
\begin{equation}
\gamma_{ \pm }  = \frac{1}{\sqrt{ 2 }} (\gamma_4 \pm \gamma_0 )
=\frac{1}{\sqrt{ 2 }}\left(\begin{array}{c|c}
\pm I_2  & iI _2  \\\hline i I _2 & \mp  I_2 
\end{array}\right).
\end{equation}

\providecommand{\href}[2]{#2}\begingroup\raggedright\endgroup

\end{document}